\documentclass[a4paper,USenglish]{lipics-v2016}
 
\usepackage{microtype}

\title{A Self-Stabilizing General De Bruijn Graph\footnote{This work was partially supported by the German Research Foundation (DFG) within the Collaborative Research Center "On-The-Fly Computing" (SFB 901).}}
\titlerunning{A Self-Stabilizing General De Bruijn Graph} 

\author[1]{Michael Feldmann}
\author[2]{Christian Scheideler}
\affil[1]{Paderborn University, F\"urstenallee 11, 33102 Paderborn, Germany\\
  \texttt{michael.feldmann@upb.de}}
\affil[2]{Paderborn University, F\"urstenallee 11, 33102 Paderborn, Germany\\
  \texttt{scheideler@upb.de}}
\authorrunning{M. Feldmann and C. Scheideler} 

\Copyright{Michael Feldmann and Christian Scheideler}


\EventEditors{}
\EventNoEds{2}
\EventLongTitle{}
\EventShortTitle{}
\EventAcronym{}
\EventYear{2017}
\EventDate{July 11, 2017}
\EventLocation{Little Whinging, United Kingdom}
\EventLogo{}
\SeriesVolume{1}
\ArticleNo{1}

\theoremstyle{plain}
\newtheorem{fact}[theorem]{Fact}
\usepackage[noend]{algpseudocode}
\usepackage{algorithm}
\usepackage{mathtools}
\DeclareMathOperator*{\argmin}{argmin} 
\DeclareMathOperator*{\argmax}{argmax} 

\begin{document}

\maketitle

\begin{abstract}
Searching for other participants is one of the most important operations in a distributed system.
We are interested in topologies in which it is possible to route a packet in a fixed number of hops until it arrives at its destination.
Given a constant $d$, this paper introduces a new self-stabilizing protocol for the $q$-ary $d$-dimensional de Bruijn graph ($q = \sqrt[d]{n}$) that is able to route any search request in at most $d$ hops w.h.p., while significantly lowering the node degree compared to the clique: We require nodes to have a degree of $\mathcal O(\sqrt[d]{n})$, which is asymptotically optimal for a fixed diameter $d$.
The protocol keeps the expected amount of edge redirections per node in $\mathcal O(\sqrt[d]{n})$, when the number of nodes in the system increases by factor $2^d$.
The number of messages that are periodically sent out by nodes is constant.
\end{abstract}

\section{Introduction}
The Internet becomes more and more relevant for every part of our society, as people increasingly use it to interact with each other and exchange information.
Common examples are real-time applications like streaming platforms, multiplayer games or social media networks that are maintained by overlay networks.
The performance of these kind of systems benefits from a low latency/delay.
For example, experiments in \cite{JB09} show that users issue fewer search requests when the latency on Google web servers is increased by only $100$ms.
For many systems there are hard deadlines on the delay that are acceptable: Multiplayer games often require server-side delays only up to $10$ms, because any higher delay would reduce the fun for the players drastically.
To keep the delay low, we require an overlay network to form a topology with a low diameter in legal states such that requests can be delivered quickly to the correct entity.
Reaching a legal state can be guaranteed if the system is \textit{self-stabilizing}, i.e., the system is able to recover from illegal states.
We are interested in self-stabilizing systems that are able to route requests to their target as fast as possible even under a large number of participants.
For example, routing in a simple line structure takes $\Theta(n)$ hops, whereas routing in a de Bruijn graph can be done in $\mathcal O(\log n)$ hops.
Both of these structures have only a constant node degree.
If the degree of the nodes is much higher, i.e., in a clique, routing can be done way more effectively: We can send requests to their destination in only one hop, since every node is connected to every other node in the system.
The drawback here is that nodes have to maintain a large number of outgoing edges, which may be very costly, especially in systems with many participants.
Our goal is to develop a self-stabilizing protocol for a network, in which the node degree is lower than the node degree in the clique, but still enables to route requests to their destination in a constant number of hops w.h.p.
Given a constant $d \in \mathbb{N}$, $d \geq 2$, our network has a diameter of at most $d$ (w.h.p.) in every legitimate state.
As a network topology, we use the $q$-ary $d$-dimensional de Bruijn graph ($q = \sqrt[d]{n}$), called \textit{general de Bruijn} graph, which was first presented in~\cite{Malyshev1997}.
The self-stabilizing protocol consists of a combination of sub-protocols: We combine the sorted list with the $q$-connected list, the standard de Bruijn graph and the $q$-ary de Bruijn graph. 
For the resulting structure it holds that each node has an outdegree of $\mathcal O(\sqrt[d]{n})$, which is asymptotically optimal for a fixed diameter $d$.
\subsection{Model} \label{sec:model}
We model a distributed system as a directed graph $G=(V,E)$ with $n = |V|$.
Each peer in the system is represented by a node $v \in V$.
Each node $v \in V$ can be identified by its unique reference or its unique identifier  $v.id \in \mathbb{N}$ (called \textit{ID}). 
Additionally, each node $v$ maintains local protocol-based variables and has a \textit{channel} $v.Ch$, which is a system-based variable that contains incoming messages.
The message capacity of a channel is unbounded and messages never get lost.
If a node $u$ knows the reference of some other node $v$, then $u$ can send a message $m$ to $v$ by putting $m$ into $v.Ch$.

We distinguish between two different types of \textit{actions}: The first type is used for standard procedures and has the form $\langle label \rangle(\langle parameters \rangle):\langle command \rangle$, where $label$ is the name of that action, $parameters$ defines the set of parameters and $command$ defines the statements that are executed when calling that action.
It may be called locally or remotely, i.e., every message that is sent to a node has the form $\langle label \rangle(\langle parameters \rangle)$.
The second action type has the form $\langle label \rangle :(\langle guard \rangle) \longrightarrow \langle command \rangle$, where $label$ and $command$ are defined as above and $guard$ is a predicate over local variables.
An action for some node $u$ may only be executed if its guard is $true$ or if there is a message in $u.Ch$ that requests to call the action.
In both cases, we call the action \textit{enabled}.
An action whose guard is simply $true$ is called \textsc{Timeout}.
When a node $u$ processes a message $m$, then $m$ is removed from $u.Ch$.

We define the \textit{system state} to be an assignment of a value to every node's variables and messages to each channel.
A \textit{computation} is an infinite sequence of system states, where the state $s_{i+1}$ can be reached from its previous state $s_i$ by executing an action that is enabled in $s_i$.
We call the first state of a given computation the \textit{initial state}.
We assume \textit{fair message receipt}, meaning that every message of the form $\langle label \rangle(\langle parameters \rangle)$ that is contained in some channel, is eventually processed.
Furthermore, we assume \textit{weakly fair action execution}, meaning that if an action is enabled in all but finitely many states of a computation, then this action is executed infinitely often. 
Consider the \textsc{Timeout} action as an example for this.
We place no bounds on message propagation delay or relative node execution speed, i.e., we allow fully asynchronous computations and non-FIFO message delivery.
Our protocol does not manipulate node identifiers and thus only operates on them in \emph{compare-store-send} mode, i.e., we are only allowed to compare node IDs to each other, store them in a node's local memory or send them in a message.
Note that we compute the hash value of a node's identifier in our protocol, but this does not manipulate the ID itself.

We are interested in the formation and maintenance of a certain graph topology (which we introduce in Section \ref{sec:prelmininaries:base_construction}) for the nodes in the distributed system.
In this paper we assume that there are no corrupted IDs in the initial state of the system, otherwise we would require failure detectors to identify corrupted IDs, which exceed the scope of this paper.
Thus we can assume that node IDs are always correct in all states, as our protocol is compare-store-send.
Nevertheless, node channels may contain an arbitrary amount of messages containing false information in initial states: We call these messages \textit{corrupted} and we will argue that all corrupted messages will eventually be processed by our protocol.
We say the system is in a \textit{legitimate (stable) state}, if the nodes and the edges form the desired graph topology and there are no corrupted messages in the system.
We are now ready to define what it means for a protocol to be self-stabilizing:

\begin{definition}[Self-stabilization] \label{self_stabilization}
A protocol is \emph{self-stabilizing} if it satisfies the following two properties:
\begin{description}
	\item[-] Convergence: Starting from an arbitrary system state, the protocol is guaranteed to arrive at a legitimate state.
	\item[-] Closure: Starting from a legitimate state, the protocol remains in legitimate states thereafter.
\end{description}
\end{definition}

There is a directed edge $(u,v) \in E$, if $u$ stores the reference of $v$ in its local memory or if there is a message in $u.Ch$ carrying the reference of $v$.
In the former case, we call that edge \textit{explicit} and in the latter case we call that edge \textit{implicit}.
In order for our distributed algorithms to work, we require the directed graph $G$ containing all explicit and implicit edges to stay at least weakly connected at every point in time.
A directed graph $G=(V,E)$ is \textit{weakly connected}, if the undirected version of $G$, namely $G'=(V,E')$ is connected, i.e., for two nodes $u,v \in V$ there is a path from $u$ to $v$ in $G'$.
Once there are multiple weakly connected components in $G$, these components cannot be connected to each other anymore as it has been shown in~\cite{DBLP:journals/tcs/NorNS13} for compare-store-send protocols.
For a graph that contains multiple weakly connected components, our protocol converts each of these components to our desired topology.
Nodes may initiate search requests at any point in time.
If node $v$ initiates a search request, it enables the action \textsc{Search}($t$), where $t \in \mathbb{N}$ is the ID of the node to be searched.
We do not assume that there is always a node with ID $t$ in the system, i.e., either the search request eventually reaches $u \in V$ with $u.id = t$, or it reaches a node at which our routing algorithm outputs "Failure!".
In both cases the routing algorithm \emph{terminates}.

\subsection{Related Work}
Peer-to-Peer Overlays that are able to route requests in one hop~\cite{DBLP:conf/hotos/GuptaLR03} or two hops~\cite{DBLP:conf/nsdi/GuptaLR04} to the target have already been proposed.
Another protocol that provides fast, but sometimes suboptimal routing as well as handling of path outages, is the \textit{Resilient Overlay Network} (RON)~\cite{DBLP:journals/ccr/AndersenBKM02}.
However, neither of the above protocols are truly self-stabilizing.

The concept of self-stabilizing algorithms for distributed systems goes back to the year 1974, when E. W. Dijkstra introduced the idea of self-stabilization in a token-based ring \cite{DBLP:journals/cacm/Dijkstra74}.
People came up with self-stabilizing protocols for various types of overlays, like sorted lists~\cite{DBLP:conf/alenex/OnusRS07}, rings~\cite{DBLP:conf/p2p/ShakerR05}, spanning trees~\cite{DBLP:conf/fsttcs/AggarwalK93}, Chord graphs~\cite{DBLP:journals/mst/KniesburgesKS14}, Skip graphs~\cite{DBLP:journals/tcs/ClouserNS12} and many more.
A self-stabilizing protocol for the clique has been presented in \cite{DBLP:journals/tcs/KniesburgesKS15}.
There is even a universal approach, which is able to derive self-stabilizing protocols for several types of topologies \cite{DBLP:journals/tcs/BernsGP13}.

In addition to the general de Bruijn graph, this paper also makes use of the standard de Bruijn graph \cite{dB1946}, for which there already exists a self-stabilizing protocol by Richa~{\it et~al} \cite{DBLP:conf/sss/RichaSS11}.
It uses the same technique as our work, namely the \textit{continuous-discrete
approach}, which was originally introduced by Naor and Wieder \cite{DBLP:journals/talg/NaorW07}.
However, the protocol in \cite{DBLP:conf/sss/RichaSS11} uses several virtual nodes per real node in order to be able to locally perform a de Bruijn hop, which works for them, because the node degree is constant in the standard de Bruijn graph.
Since nodes have a degree of $\mathcal O(\sqrt[d]{n})$ in our system, we use a different approach here.
\subsection{Our Contribution}
In this paper we propose a new self-stabilizing protocol \textsc{BuildQDeBruijn} for the general de Bruijn graph, which is built out of a combination of sub-protocols.
We describe this protocol in Section~\ref{sec:desc}.
Routing packets in our network can be done in at most $d$ hops w.h.p. for any constant $d$ (Section~\ref{sec:preliminaries}).
We show that our protocol is self-stabilizing in Section~\ref{sec:analysis} among some further properties: Each node has a degree of $\mathcal O(\sqrt[d]{n})$ and only sends out a constant number of messages in each call of \textsc{Timeout}.
Also, if the number of nodes increases by a factor of $2^d$, each old node only has to redirect or build at most $\mathcal O(\sqrt[d]{n})$ edges on expectation.

\section{Topology and Routing}\label{sec:preliminaries}
In this section we introduce our construction for emulating a general de Bruijn graph.
We also describe how to route search requests via this construction and show that the routing algorithm for each request performs at most $d$ hops w.h.p. until termination.

\subsection{Classical De Bruijn Graphs and Hashing}\label{sec:hashing}
The classical de Bruijn graph is defined as follows:

\begin{definition}
Let $d \in \mathbb{N}$.
The standard ($d$-dimensional) de Bruijn graph consists of nodes $(x_1, \ldots ,x_d) \in \{0,1\}^d$ and edges $(x_1, \ldots ,x_d) \rightarrow (j,x_1, \ldots ,x_{d-1})$ for all $j \in \{0,1\}$.
\end{definition}

The standard de Bruijn graph has a diameter of $d$, so one can route a packet from a source $s \in \{0,1\}^d$ to a target $t \in \{0,1\}^d$ by adjusting exactly $d$ bits.
We call one single bitshift a \textit{de Bruijn hop}.
If we assume $d$ to be a constant, then the number of hops per search request is constant.
However, the standard de Bruijn graph has a fixed number of nodes in this case, that is, $n = 2^d$.
Since we want to allow an arbitrary number of nodes in the system, the standard de Bruijn graph does not fit our purposes.
Therefore, we extend the standard de Bruijn graph to the general de Bruijn graph, which is defined as follows:

\begin{definition}\label{def:general_db}
Let $q,d \in \mathbb{N}$.
The general ($q$-ary $d$-dimensional) de Bruijn graph consists of nodes $(x_1, \ldots ,x_d) \in \{0, \ldots ,q-1\}^d$ and edges \[(x_1, \ldots ,x_d) \rightarrow (j,x_1, \ldots ,x_{d-1})\] for all $j \in \{0, \ldots ,q-1\}$.
\end{definition}

The diameter of the general de Bruijn graph is also $d$, so we are still able to route search requests in $d$ hops by adjusting exactly $d$ bits.
We allow $q$ to be dynamic, so we can use this topology to maintain any number of nodes, that is, $n = q^d$.
Solving this equation for $q$ yields a degree of $q = \sqrt[d]{n}$ per node.
Thus, the general de Bruijn graph meets the following lower bound:

\begin{fact} \label{fact:node_degree_lower_bound}
Every graph with $n$ nodes and diameter $d$ must have a degree of at least $\lfloor \sqrt[d]{n} \rfloor$.
\end{fact}

\begin{proof}
	Assume to the contrary for a graph with $n$ nodes and diameter $d$ that no nodes has a degree higher than $\lfloor \sqrt[d]{n} \rfloor -1$.
	Fix a node $v$ and construct the BFS-tree starting at $v$ with $d$ until level $d$.
	The number of leaf nodes in this tree is equal to $(\lfloor \sqrt[d]{n} \rfloor-1)^d < (\lfloor \sqrt[d]{n} \rfloor)^d \leq n$, which implies that we cannot reach all nodes from $v$ in just $d$ steps.
\end{proof}

We use a pseudorandom hash function $h: \mathbb{N} \rightarrow [0,1)$ to distribute node IDs uniformly and independently onto the $[0,1)$-interval.
Whenever we want to use the hash value of a node $v \in V$, we just write $v$ instead of $h(v.id)$ for convenience.
We can derive a bit string representation of the first $k$ bits out of a node $v$'s hash value by computing the inverse of the function $r_k: \{0,1\}^k \rightarrow [0,1)$ with \[r_k(x_1,\ldots,x_k) = \sum_{i=1}^k x_i\cdot \frac{1}{2^i}.\]
Once we have a bit string representation of a node, we can transform it to any base $q = 2^k$ for some $k \in \mathbb{N}$, $k > 1$.
Both of these transformations are important for our routing algorithm.

A node $u$ is \textit{left} (resp. \textit{right}) of a node $v$, if $u < v$ (resp. $u > v$).
Given some node $w$ and two nodes $u, v$, we say that $u$ is \textit{closer} to $w$ than $v$, if $|u-w| < |v-w|$.
We call a node $u \neq v$ the \textit{closest neighbor} of $v \in V$, if there are no other nodes that are closer to $v$ than $u$.
Similarly, a node $v$ is \textit{closest} to some point $p \in [0,1)$, if $|v-p| < |u - p|$ for all $u \in V, u \neq v$.
For a hash function $h$ as described above, we get the following lemma:

\begin{lemma}\label{cor:nodedistance}
The expected distance between two closest neighbors $u, v \in V$ on the $[0,1)$-interval~(seen as a ring) is equal to $\frac{1}{n}$, where $n$ denotes the number of nodes in the system.
\end{lemma}

\begin{proof}
Let $X_1,\ldots,X_n$ be the hashes of $n$ nodes when using $h$.
W.l.o.g. assume that there exists $i \in \{1,\ldots,n\}$ with $X_i = 0$.
The cumulative distribution function (CDF) $F_{X_j}(x)$ is defined by \[F_{X_j}(x) = Pr[X_j \leq x]\] for an arbitrary $j \in \{1,\ldots,n\}$.
It holds that the CDF of the minimum value $X_{\min} := \min\{\{X_1,\ldots,X_n\}\setminus \{X_i\}\}$ is given by \[F_{X_{\min}}(x) = 1-(1-Pr[X_{\min} \leq x])^{n-1}.\]
Since 
\begin{equation*}
   Pr[X_j \leq x] =
   \begin{cases}
     0, & x < 0 \\
     x, & 0 \leq x \leq 1 \\
     1, & x > 1 
   \end{cases}
\end{equation*}
it follows 
\begin{equation}
   F_{X_{\min}}(x) =
   \begin{cases}
     0, & x < 0 \\
     1-(1-x)^{n-1}, & 0 \leq x \leq 1 \\
     1, & x > 1 
   \end{cases}
\end{equation}
We need to compute the expected distance between $X_i = 0$ and $X_{\min}$, since these hashes belong to neighboring nodes.
It holds
\begin{eqnarray*}
E[|X_{min} - X_i|] = E[X_{\min}] &=& \int_0^{\infty} 1-F_{X_{\min}}(t)\ dt\\
&\overset{(1)}{=}&  \int_0^{1} 1-(1-(1-t)^{n-1})\ dt\\
&=&  \int_0^{1} (1-t)^{n-1}\ dt = \left[-\frac{1}{n}(1-t)^n\right]_0^1 = \frac{1}{n}
\end{eqnarray*}
\end{proof}

For the rest of this paper, we require $h: \mathbb{N} \rightarrow [0,1)$ and the constant $d \in \mathbb{N}$ to be a part of our protocol, i.e., every node knows $h$ and $d$.

\subsection{Base Construction}\label{sec:prelmininaries:base_construction}
We hash all nodes onto the $[0,1)$-interval, using the hash function $h$ as described in the last section.
The network we are going to construct has a diameter of $d$ w.h.p., which makes routing in a constant number of hops possible.

\begin{definition}[Network Topology]\label{def:topology}
The \emph{general de Bruijn network (GDB)} is a directed graph $G=(V,E_L \cup E_q \cup E_{dB} \cup E_{q-dB})$ with the following properties:
	\begin{description}
		\item[-] $E_L$ contains list edges: $(v,w) \in E_L \Leftrightarrow w$ is the closest neighbor that is left (resp. right) of $v$.
		\item[-] $E_q$ contains $q$-neighborhood edges: $(v,w) \in E_q \Leftrightarrow$ there are at most $c\cdot q$ nodes closer to $v$ than $w$, where $c > 2$ is a constant and $q = \sqrt[d]{n}$.
		\item[-] $E_{dB}$ contains standard de Bruijn edges: $\forall j \in \{0,1\}: (v,w) \in E_{dB} \Leftrightarrow w$ is closest to the point $\frac{v+j}{2}$.
		\item[-] $E_{q-dB}$ contains general de Bruijn edges: $\forall i \in \{2,\ldots,\log(q)\}\ \forall j \in \{0,\ldots,2^i-1\}: (v,w) \in E_{q-dB} \Leftrightarrow w$ is closest to the point $\frac{v+j}{2^i}$.
	\end{description}
\end{definition}

All logarithms in this paper are to the base $2$.
For the natural logarithm of some number $x$ we use $\ln(x)$.
Note that the constant $c > 2$ is only needed to prove the correctness of the routing algorithm.
Assume for simplicity that $q = \sqrt[d]{n}$ is a power of $2$, i.e., $q = 2^k$ for some $k \in \mathbb{N}$.
We explain how to deal with arbitrary values of $q$ in Section~\ref{q_v_neighborhood}.

If $w$ is closest to the point $\frac{v+j}{2^i}$, denote the edge $(v,w)$ as a \textit{de Bruijn edge on level $i$}.
For $i = 1$, we speak of a \textit{standard de Bruijn edge}.
For $i > 1$, we speak of a \textit{general de Bruijn edge} and if $i < \log(q)$ we speak of a \textit{lower level general de Bruijn edge}.
Note that we include lower level general de Bruijn edges to facilitate the self-stabilization process.
If we forward a message via a de Bruijn edge on level $i>1$, we speak of a \textit{general de Bruijn hop}.
For $i = 1$, we speak of a \textit{standard de Bruijn hop}.
By writing $v \rightarrow p$ for a point $p \in [0,1)$, we mean that $v$ has an edge to the node $u$ that is closest to $p$, i.e., $v$ stores the reference of $u$ in its local memory.
We are now ready to prove that de Bruijn edges in our network emulate the classical de Bruijn edges correctly:

\begin{lemma}\label{lemma:dbhop_correctness}
Let $v \in V$. 
A de Bruijn hop via $v \rightarrow \frac{v+j}{2^{i}}$, $i \in \{1,\ldots,\log(q)\}$, $j \in \{0, \ldots ,2^{i} - 1\}$, is equivalent to appending $\log(2^{i}) = i$ bits to the left of the bit string representation of $v$, where the content of the appended bit string is equal to $(b_{i-1},b_{i-2}, \ldots ,b_0) \in \{0,1\}^i$ with $b_{i-1}\cdot 2^{i-1} + b_{i-2}\cdot 2^{i-2} +  \ldots  + b_{0}\cdot 2^{0} = j$.
\end{lemma}

\begin{proof}
We can write $v$ as \[v = a_1\cdot \frac{1}{2} + a_2\cdot \frac{1}{4} +  \ldots  + a_m\cdot\frac{1}{2^m}\] with $a_k \in \{0,1\}$ for all $k \in \{1, \ldots ,m\}$.
Then it holds 
\begin{eqnarray*}
\frac{v+j}{2^{i}} &=& \frac{v}{2^{i}} + \frac{j}{2^{i}}\\
&=& \frac{a_1\cdot \frac{1}{2} + a_2\cdot \frac{1}{4} +  \ldots  + a_m\cdot\frac{1}{2^m}}{2^{i}} + \frac{j}{2^{i}}\\
&=& \underbrace{\frac{a_1}{2^{i+1}} + \frac{a_2}{2^{i+2}} +  \ldots  + \frac{a_m}{2^{i + m}}}_\text{Bits of $v$ shifted $i$ times to the right} + \frac{j}{2^{i}}\\
\end{eqnarray*}
We know that $j \in \{0, \ldots ,2^{i}-1\}$, so we can write $j$ as a binary string with $i$ bits: \[j = b_{i-1}\cdot 2^{i-1} + b_{i-2}\cdot 2^{i-2} +  \ldots  + b_{0}\cdot 2^{0}\] with $b_l \in \{0,1\}$ for all $l \in \{0, \ldots ,i-1\}$.
So we have
\begin{eqnarray*}
\frac{v+j}{2^{i}} &=& \frac{a_1}{2^{i+1}} + \frac{a_2}{2^{i+2}} +  \ldots  + \frac{a_m}{2^{i + m}} + \frac{j}{2^{i}}\\
&=& \frac{a_1}{2^{i+1}} + \frac{a_2}{2^{i+2}} +  \ldots  + \frac{a_m}{2^{i + m}} + \frac{b_{i-1}\cdot 2^{i-1} + b_{i-2}\cdot 2^{i-2} +  \ldots  + b_{0}\cdot 2^{0}}{2^{i}}\\
&=& \underbrace{\frac{b_{i-1}}{2} + \frac{b_{i-2}}{4} +  \ldots  + \frac{b_0}{2^{i}}}_\text{$i$ bits defined by $j$ appended to the left} + \underbrace{\frac{a_1}{2^{i+1}} + \frac{a_2}{2^{i+2}} +  \ldots  + \frac{a_m}{2^{i + m}}}_\text{Bits of $v$ shifted $i$ times to the right}\\
\end{eqnarray*}
Therefore, we have appended $i$ bits to $v$'s bit string and proved the lemma.
\end{proof}

Since $j \in \{0, \ldots ,2^i-1\}$, we are able to append any arbitrary bit string of length $i$.
So for $i = \log(q)$, we can append $\log(q) = \log(\sqrt[d]{n}) = \frac{1}{d}\log(n)$ arbitrary bits at once per general de Bruijn hop.
The outdegree of our construction is not too high as the following theorem states:

\begin{theorem} \label{theorem_degree}
Each node in the GDB has degree $\mathcal O(\sqrt[d]{n})$.
\end{theorem}

\begin{proof}
We count the maximum number of edges a node may possibly have:
Each node $v \in V$ maintains the $c\cdot q$ closest neighbors of $v$, which already includes $v$'s direct list neighbors.
Next count the number of general de Bruijn edges for $v$, including $v$'s standard de Bruijn edges: On level $i$, $i \in \{1, \ldots ,\log(q)\}$, there are exactly $2^{i}$ de Bruijn edges, leading from $q$ general de Bruijn edges on level $\log(q)$ to the $2$ standard de Bruijn edges on level $1$. So $v$ has
\begin{align*}
& \sum_{i=1}^{\log(q)} 2^{i}\\
=\ & \left(\sum_{i=0}^{\log(q)} 2^i\right) - 1\\
=\ & \left(\frac{1-2^{\log(q)+1}}{1-2}\right)-1\\
=\ & (2q-1)-1\\
=\ & 2q - 2\\
\end{align*}
general de Bruijn edges.

Summing it all up results in $v$ having $cq + 2q - 2 = (c+2)q-2 = \mathcal O(\sqrt[d]{n})$ outgoing edges.
\end{proof}

\subsection{Routing}\label{sec:pre:routing}
When processing a search request with target ID $t \in \mathbb{N}$, we proceed in two phases: In the first phase we perform $d-1$ general de Bruijn hops to fix the most significant bits of the target address.
In the second phase, we greedily search for the target node via $q$-neighborhood edges.

At the beginning of the first phase, we compute the bit string representation of $h(t)$ and transform it to the base $q$ as described in Section~\ref{sec:hashing}.
This yields a number $t_q := (t_1,\ldots,t_k)_q \in \{0,\ldots,q-1\}^k$ for some $k \in \mathbb{N}$.
We only consider the first $d-1$ digits $t_1,\ldots,t_{d-1}$ of $t_q$.
Let the search request be at node $v \in V$.
We perform a general de Bruijn hop via the edge $v \rightarrow \frac{v+t_i}{q}$ starting with $i = d - 1$.
We decrement $i$ after each general de Bruijn hop.
The first phase ends, when $i = 0$, i.e., after $d-1$ general de Bruijn hops.
Observe that at this point, we have fixed the most significant $\lceil \frac{d-1}{d}\log(n) \rceil$ bits of the bit string representation of $h(t)$.

In the second phase, we greedily search for the node with target ID $t$, by delegating the search request via edges in $E_q$.
We do this until the target node has been found, or the request arrives at a node $v \in V, v.id \neq t$ from which it cannot be routed closer to $h(t)$ via $q$-neighborhood edges.
In both cases, the algorithm terminates, resulting in a successful search in the first case or a failed search in the second case.
This phase is equivalent to fixing the remaining bits of the binary representation of $h(t)$, which can be done via a single hop w.h.p. until the request arrives at the target node. 
Figure~\ref{fig:routing_example} illustrates an example when the constant $d$ is set to $4$.

\begin{figure}[ht]
	\centering
 	\includegraphics[scale=0.99]{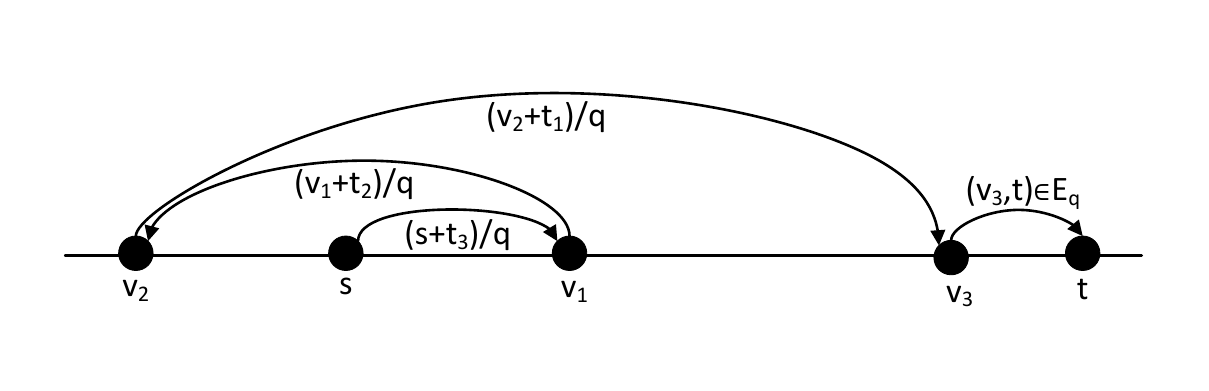}
	\caption{Possible routing path to node $t$ starting at node $s$, when $d = 4$.}
	\label{fig:routing_example}
\end{figure}

To show the correctness of the routing algorithm, we need the following theorem for the standard Chernoff bounds:

\begin{theorem}[Chernoff Bounds]\label{theorem:chernoff}
Let $X_1,\ldots,X_n$ be independent random variables.
Let $X = \sum^n_{i=1} X_i$ and  $\mu = E[X]$.
Then it holds for all $\delta > 0$ that 
\begin{equation}\tag{1}
Pr[X \geq (1+\delta)\mu] \leq e^{-\delta^2\mu/(2( 1+\delta/3))}
\end{equation}
and for all $0 < \delta < 1$ that 
\begin{equation}\tag{2}
Pr[X \leq (1-\delta)\mu] \leq e^{-\delta^2\mu/2}.
\end{equation}
\end{theorem}

The following theorem yields the desired bound on the number of hops for the routing algorithm:

\begin{theorem} \label{theorem:routing}
The number of hops required to send a request from a source node $s$ to a destination node $t$ via \textsc{DeBruijnSearch} is $d$ w.h.p.
\end{theorem}

\begin{proof}
Let $(t_1,\ldots,t_k)_2$ be the bits of the target destination $h(t)$ for some $k \in \mathbb{N}$.
In the first phase of \textsc{DeBruijnSearch}, we perform $d-1$ general de Bruijn hops.
Lemma~\ref{lemma:dbhop_correctness} implies that we arrive at some node $v$ with bit address $(v_1,\ldots,v_l)_2$, $l \in \mathbb{N}$ and \[v_i = t_i\ \forall\ i \in \left\{1,\ldots,\left\lceil\frac{d-1}{d}\log(n)\right\rceil\right\}.\]
Assume $k = l$ for convenience.
It holds $remHops = 0$ at this point, so our algorithm switches to the second phase.
The remaining bits that need to be fixed are the bits $t_{\lceil\frac{d-1}{d}\log(n)\rceil+1},\ldots,t_k$.
We show that these bits can be fixed in only one hop via the $q$-neighborhood w.h.p., because in worst case it holds \[v_i \neq t_i\ \forall\ i \in \left\{\left\lceil\frac{d-1}{d}\log(n)\right\rceil + 1,\ldots,k\right\},\] so the maximum distance between $v$ and $h(t)$ on the $[0,1)$-interval is equal to

\begin{eqnarray*}
\sum_{i=\lceil\frac{d-1}{d}\log(n)\rceil+1}^{k} \frac{1}{2^i} &\leq & \underbrace{\left(\sum_{i=1}^{k} \frac{1}{2^i}\right)}_{\leq 1\text{ for } k \rightarrow \infty} \cdot \frac{1}{n^{\frac{d-1}{d}}}\\
&\leq & \frac{1}{n^{\frac{d-1}{d}}} = \frac{\sqrt[d]{n}}{n} = \frac{q}{n}\\
\end{eqnarray*}

We need to show that $t \in v.Q$ with high probability, when $|t-v| \leq \frac{q}{n}$.
This is equivalent to showing that the probability of $c\cdot q$ or more nodes being in $I = [v-\frac{q}{n}, v+ \frac{q}{n}]$ is low.
For all $w \in V$ let $X_w$ be a binary random variable with
\begin{equation*}
   X_w =
   \begin{cases}
     1, & \text{if }w \in I \\
     0, & \text{otherwise.} 
   \end{cases}
\end{equation*}
Then it holds $Pr[X_w = 1] = 2q/n$ and $E[X_w] = 2q/n$.
Define \[X:=\sum^{}_{w \in V} X_w.\]
Then $\mu = E[X] = \sum^{}_{w \in V} E[X_w] = n\cdot \frac{2q}{n} = 2q$.
Following the standard Chernoff bound (Theorem~\ref{theorem:chernoff}(1)) with $\delta = c/2-1$, we get
\begin{eqnarray*}
Pr[X \geq c\cdot q] & \leq & e^{(-(c/2-1)^2 \cdot 2q)/(2((1+c/2-1)/3))}\\
& = & e^{(-(c/2-1)^2 \cdot 2q)/(c/3)}\\
& = & e^{-pq} \text{, for } p := (6(c/2-1)^2)/c\\
& = & e^{-p\sqrt[d]{n}}\\
& \leq & e^{-p \ln(n)} \text{, for } n \text{ high enough}\\
& = & n^{-p}
\end{eqnarray*}
\end{proof}

Notice that Theorem~\ref{theorem:routing} still holds when $q$ is not exactly accurate but only a value in $\Theta(\sqrt[d]{n})$, because $\ln(n) \in \Theta(\sqrt[d]{n})$.
This is important, because our self-stabilizing protocol in the next section uses approximations of $q$, resp. $\log(n)$.

\section{The BuildQDeBruijn Protocol} \label{sec:desc}
In this section we describe the \textsc{BuildQDeBruijn} protocol.
We construct the protocol out of sub-protocols for each edge type mentioned in Definition~\ref{def:topology}.
The pseudocode can be found in Appendix~\ref{appendix:pseudocode}.

\subsection{Node Variables} \label{sec:desc:variables}
We first give an overview over the variables of each node:

\begin{definition}
Given a GDB $G$, each node $v \in V$ has the following variables:
\begin{description}
	\item[-] Variables $v.left, v.right \in V \cup \{\perp\}$ storing $v$'s left and right list neighbor.
	\item[-] A variable $v.q \in 2^k$, $k \in \mathbb{N}$ storing an approximation of $\frac{1}{2}\sqrt[d]{n}$.
	\item[-] A set $v.Q := \{q_1, \ldots ,q_{c \cdot 2v.q}\} \subset V$ storing nodes for $v$'s $q$-neighborhood.
	\item[-] Variables $v.db(i,j) \in V \cup \{\perp\}$, for all $i \in \{1, \ldots ,\log(2v.q)\}$, $j \in \{0, \ldots ,2^i - 1\}$ representing $v$'s de Bruijn edges. Denote the union of $v$'s de Bruijn edges by the set $v.db = \bigcup_{i,j} v.db(i,j)$
\end{description}
\end{definition}

Observe that $v.db(1,0)$ and $v.db(1,1)$ represent $v$'s standard de Bruijn edges.
If our protocol has to call an action on a node stored in variable $u$, it only executes this call, if $u \neq \perp$.
\textsc{BuildQDeBruiijn} consists of four sub-protocols: One for list edges, one for $q$-neighborhood edges, one for standard de Bruijn edges and a sub-protocol for general de Bruijn edges.
We describe each sub-protocol individually in the following sections.

\subsection{List Edges}
The base of our self-stabilizing protocol consists of a sorted list for all nodes $v \in V$ over the $[0,1)$-interval.
We use the \textsc{BuildList} protocol from \cite{DBLP:conf/alenex/OnusRS07}, where each node only keeps its closest left ($v.left$) and right ($v.right$) list neighbor.
In every call of \textsc{Timeout}, each node introduces itself to $v.left$ and $v.right$, by sending a \Call{Linearize}{$v$} request to them.
When calling \Call{Linearize}{$v$} on a node $u$, $u$ sets $u.left = v$, if $v$ is left of $u$ and closer to $u$ than $u.left$.
The old value $o$ of $u.left$ is then delegated to the node $\bar{q} \in u.Q$ that is closest to $o$ by calling \Call{Linearize}{$o$} on $\bar{q}$.
In case $u.left = \perp$, $u$ just sets $u.left = v$.
If $v$ is left of $u$ and $u.left$ is closer to $u$ than $v$, then $u$ delegates $v$ as described above.
Node $u$ proceeds analogously for $u.right$ in case $v$ is right of $u$.
Thus, node references are never deleted, but always delegated until the node arrives at the correct spot in the sorted list.
We get the following theorem from~\cite{DBLP:conf/alenex/OnusRS07}:

\begin{theorem} [\cite{DBLP:conf/alenex/OnusRS07}]\label{list:self_stabilizing}
\textsc{BuildList} is self-stabilizing:
\begin{description}
	\item[-] Convergence: \textsc{BuildList} converts any weakly connected graph $G = (V, E_L)$ into a sorted list.
	\item[-] Closure: If the explicit list edges in $G = (V,E_L)$ already form a sorted list, then these edges are preserved by \textsc{BuildList}.
\end{description}

\end{theorem}

Theorem~\ref{list:self_stabilizing} does not suffice to guarantee convergence for the sorted list in our protocol because we just require $G = (V, E_L \cup E_q \cup E_{dB} \cup E_{q-dB})$ to be weakly connected.
Therefore, we \emph{downgrade} (non-list) edges represented by sets $v.Q$ and $v.db$, if they are closer to $v$ than $v.left$ or $v.right$: Downgrading some node $u$ stored in one of these sets is done in \textsc{Timeout} of each sub-protocol other than \textsc{BuildList}, by locally calling \textsc{Linearize}($u$).
Similarly we may \emph{upgrade} list edges represented by $v.left$ and $v.right$ in case they are a better fit w.r.t. Definition~\ref{def:topology} than nodes stored in sets $v.Q$ and $v.db$.
Upgrading is done by copying the node reference from $v.left$, resp. $v.right$ and storing the copy in $v.Q$ or $v.db$.
Figure~\ref{fig:protocol_dependencies} illustrates the interaction between sub-protocols of \textsc{BuildQDeBruijn}.

\begin{figure}[ht]
	\centering
 	\includegraphics[scale=0.76]{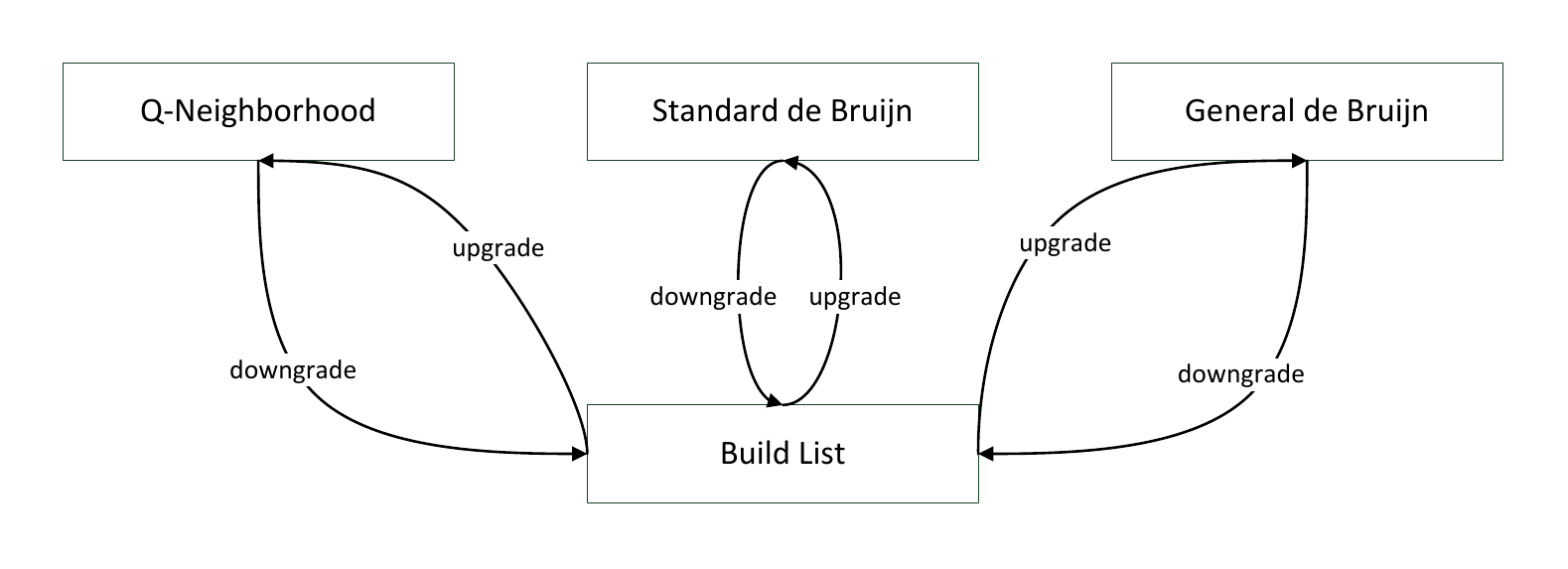}
	\caption{Interaction between all sub-protocols of \textsc{BuildQDeBruijn}.}
	\label{fig:protocol_dependencies}
\end{figure}

\subsection{Q-Neighborhood} \label{q_v_neighborhood}
Every node $v \in V$ needs to keep edges to its closest $c\cdot q = c\sqrt[d]{n}$ neighbors.
Since $v$ is not able to determine the exact value of $\sqrt[d]{n}$ locally, it stores an approximation in its variable $v.q$.
Instead of aiming for $v.q \approx \sqrt[d]{n}$, we aim for $v.q \approx\frac{1}{2}\sqrt[d]{n}$ for convergence reasons.
Whenever we want to use the (approximated) value $\sqrt[d]{n}$ at $v$, we just use $v.q$ multiplied by $2$.
If $v$ modifies $v.q$, we call this event a \textit{$v.q$-update}.
Using $v.q$, $v$ maintains the set $v.Q := \{q_1, \ldots ,q_{c\cdot 2\cdot v.q}\} \subset V$ storing the $c\cdot 2\cdot v.q$ nodes closest to $v$.
For $v.q \approx \frac{1}{2}\sqrt[d]{n}$, it holds $|v.Q| \approx c\cdot \sqrt[d]{n}$.
As soon as the system is in a legitimate state, it holds for any node $u \neq v$ with $u \not \in v.Q$ that $|u-v| > \max_{i \in \{1, \ldots ,c\cdot 2\cdot v.q\}}\{|q_i-v|\}$, i.e., $v.Q$ contains $v$'s closest $c \cdot \sqrt[d]{n}$ list neighbors.
Next we describe how our protocol updates $v.Q$ and $v.q$. 

To keep $v.Q$ updated at any time, $v$ does the following: In each call of \textsc{Timeout}, $v$ picks $q_k \in v.Q$ in a round-robin fashion and introduces $q_k$ to its closest list neighbor in the direction of $v$ by calling \Call{Introduce}{$\tilde{q}$, $v$} on $q_k$. 
The node $\tilde{q}$ is determined as follows: If $q_k = v.left$ or $q_k = v.right$, then $\tilde{q} = v$. 
Otherwise, $v$ sets $\tilde{q}$ based on $q_k$ being left or right of $v$: If $q_k < v$, then $\tilde{q} = q_{k+1}$, otherwise $\tilde{q} = q_{k-1}$.

When some node $u$ receives an \Call{Introduce}{$\tilde{q}$, $v$} request, $u$ updates $u.Q$ by choosing the closest $c\cdot 2\cdot u.q$ neighbors from $u.Q \cup \{\tilde{q}\}$.
Nodes $\bar{q} \in u.Q \cup \{\tilde{q}\}$ that are not part of the updated set $u.Q$ are delegated via the \textsc{BuildList} protocol by locally calling \textsc{Linearize}($\bar{q}$).
Afterwards, $u$ responds by sending an \Call{Introduce}{$\{l\}$, $\perp$} message to $v$, where $l = u.left$, if $u.right$ is closer to $v$ than $u.left$, or $l = u.right$ otherwise.
This has to be done in order to guarantee that every node $v$ eventually has a complete set $v.Q$ with $|v.Q| = c\cdot 2\cdot v.q$.
Note that the second parameter is set to $\perp$ for this response, in order to avoid an infinite loop of message calls between two nodes.

\begin{figure}[ht]
	\centering
 	\includegraphics[scale=1]{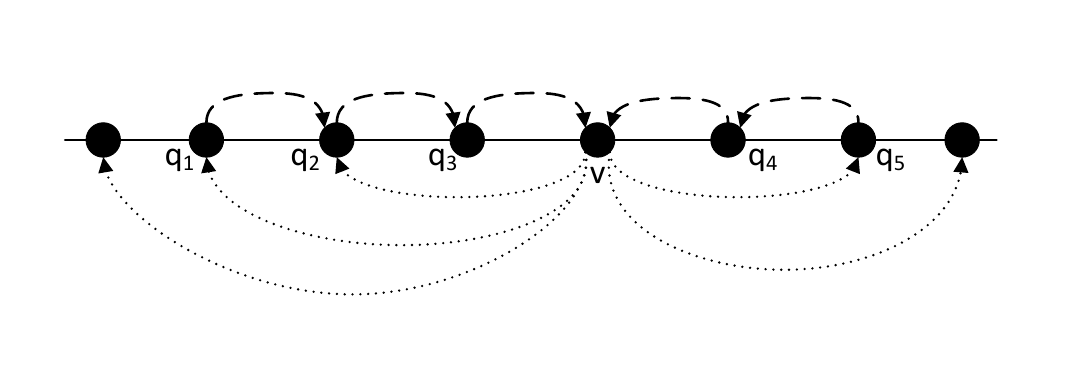}
	\caption{Implicit edges generated after $v$ has chosen $q_1,\ldots,q_5$ once in \textsc{Timeout}. The dotted implicit edges are generated by the responses sent out from $q_1,\ldots,q_5$ to $v$.}
	\label{fig:introduction_illustration}
\end{figure}

To keep $v.q$ updated at node $v$, $v$ periodically checks if $v.q$ is within the interval $(\frac{1}{4}\sqrt[d]{n}, \sqrt[d]{n})$.
Recall that we require $v.q$ to be a power of $2$, i.e., $v.q = 2^k$ for some $k \in \mathbb{N}$.
If $v.q \not\in (\frac{1}{4}\sqrt[d]{n}, \sqrt[d]{n})$, it has to be updated.
Notice that we have to avoid updating $v.q$ too frequently because each update changes the set $v.Q$, implying a higher workload for $v$.
The way we approximate $v.q$ is the following: We calculate values
	\[a_i = \left|2^d \cdot |q_1-q_{2^i\cdot v.q}| - \left(\frac{1}{2^i \cdot v.q}\right)^{d-1}\right|,\]
for all $i \in \{-\log(v.q), \ldots ,0,1\}$.
Out of those $a_i$, we compute $j$ such that $a_j = \min_i\{a_i\}$ and multiply $v.q$ by $2^j$. 
As the next lemma shows, this leads to $v.q$ becoming stable, i.e., $v.q$ is not updated anymore at some point in time.

\begin{lemma}\label{q_v_approx_lemma}
Consider a sorted list over the interval $[0,1)$ and a node $v \in V$.
After at most $\log(\sqrt[d]{n})$ $v.q$-updates, $v.q \in (\frac{1}{4}\sqrt[d]{n}, \sqrt[d]{n}) = \Theta(\sqrt[d]{n})$ w.h.p. and $v.q$ does not get updated anymore as long as no nodes join or leave the system.
\end{lemma}

\begin{proof}
For a node $v \in V$ consider the function \[f(v.q) = \left| 2^d\cdot \left| q_1 - q_{v.q} \right| - \left( \frac{1}{v.q}\right)^{d-1}\right|\] with $q_1, q_{v.q} \in v.Q$ being the nodes with minimum, resp. maximum hash value.
This function $f$ represents the relation between the choice of $v.q$ and the number of nodes in the interval $I = [v-\frac{v.q}{n}, v + \frac{v.q}{n}]$.
When $v.q$ gets closer to $\frac{1}{2}\sqrt[d]{n}$, $f(v.q)$ gets smaller, which results in $v.q$ getting stable, since our approximation approach searches for a $v.q$ that minimizes $f$.
On the other hand, if $f$ grows larger, then $v.q$ gets more and more inaccurate compared to its approximation target $\frac{1}{2}\sqrt[d]{n}$.
Notice that the expected distance between nodes $q_1$ and $q_{v.q}$ is equal to $v.q/n$, which implies that the closer $f(v.q)$ is to $0$, the closer $v.q$ is to $\frac{1}{2}\sqrt[d]{n}$.

We show for a value $v.q$ that is too small, the number of nodes in the interval $I$ for $v$ is too small w.h.p., which results in $|q_1 - q_{v.q}|$ and ultimately $f(v.q)$ getting too large.
Notice that in this case $q_1$ and $q_{v.q}$ are not contained in $I$ w.h.p.
Similarly we show for a value $v.q$ that is too large, the number of nodes in the interval $I$ of $v$ is too large, which also results in $f$ getting too large.

Assume that $v.q$ is chosen too small, i.e., $v.q \leq \frac{1}{4} \sqrt[d]{n}$.
Then it holds $|I| \leq \frac{\sqrt[d]{n}}{2n}$.
For all $w \in V$ let $X_w$ be a binary random variable with
\begin{equation*}
   X_w =
   \begin{cases}
     1, & \text{if } w \in I \\
     0, & \text{otherwise.} 
   \end{cases}
\end{equation*}
Then it holds $Pr[X_w = 1] = \frac{\sqrt[d]{n}}{2n}$ and for $X := \sum_{w \in V} X_w$ it holds $\mu = E[X] = n\cdot \frac{\sqrt[d]{n}}{2n} = \sqrt[d]{n}/2$.
Following the standard Chernoff bound (Theorem~\ref{theorem:chernoff}(1)), we get for all $\delta > 0$:

\begin{eqnarray*}
Pr[X \geq (1 + \delta) \cdot \mu] & \leq & e^{-\frac{-\delta^2\sqrt[d]{n}}{4(1+\delta/3)}}\\
& = & e^{-p\sqrt[d]{n}} \text{, for } p = \frac{\delta^2}{4(1+\delta/3)}.\\
& \leq & n^{-p} \text{, for } n \text{ high enough}.
\end{eqnarray*}

Thus the number of nodes in $I$ is too small w.h.p., which results in $v.q$ getting increased by our algorithm as $f(v.q)$ is too large.

Now assume that $v.q$ is chosen too large, i.e., $v.q \geq \sqrt[d]{n}$.
Then it holds $|I| \geq 2\sqrt[d]{n}/n$.
We define the same binary random variables $X_w$ for all $w \in V$ as above and get $Pr[X_w = 1] \geq 2\sqrt[d]{n}/n$ as well as $\mu = E[X] = n \cdot \frac{2\sqrt[d]{n}}{n} = 2\sqrt[d]{n}$.
Following the standard Chernoff bound (Theorem~\ref{theorem:chernoff}(2)) results in
\begin{eqnarray*}
Pr[X \leq (1 - \delta) \cdot \mu] & = & Pr[X \leq (1- \delta) \cdot 2\sqrt[d]{n}]\\
& \leq & e^{-\frac{-\delta^2 \cdot 2\sqrt[d]{n}}{2}}\\
& = & e^{-\delta^2 \cdot \sqrt[d]{n}}\\
& \leq & n^{-\delta^2} \text{, for } n \text{ high enough and all } 0 < \delta < 1.
\end{eqnarray*}

Thus there are too many nodes in $I$ w.h.p., resulting in $v.q$ getting halved by our algorithm as $f(v.q)$ is too large.

We only need at most $\mathcal O(\log(\sqrt[d]{n}))$ $v.q$-updates until we found the correct $v.q$, because we need to multiply $v.q$ at most $\log(v.q)$ times with $2$ in case $v.q \leq \frac{1}{4}\sqrt[d]{n}$ initially.
For any $v.q \geq \sqrt[d]{n}$ we just need one $v.q$-update, because we compute $a_i$ ($i \in \{-\log(v.q),\ldots,0,1\}$) for all candidates $2^iv.q$ that are lower than $v.q$ and find the best fitting candidate in one step.
\end{proof}

Note that Lemma~\ref{q_v_approx_lemma} only holds if for a fixed value $v.q$, $v.Q$ eventually contains the correct nodes.
But this can be shown as part of the overall convergence (Lemma~\ref{lemma:convergence:II}).

In addition to the approximation of $\sqrt[d]{n}$, we need an approximation of $\log(n)$ at every node $v \in V$ in order to perform routing.
We approximate $\log(n)$ similar to the approach for computing $v.q$: For all $i \in \{\frac{1}{2}v.q, \ldots 2v.q\}$, we compute a value 	
	\[a_i = \left|2^d \cdot |q_1-q_{i}| - \left(\frac{1}{i}\right)^{d-1}\right|\]
and set $\log(n) = \log((2\cdot \argmin_i\{a_i\})^d)$, as $\argmin_i\{a_i\}$ gives us the integer value $i$ that is closest to $\frac{1}{2}\sqrt[d]{n}$.
The resulting approximation for $\log(n)$ is even more precise than the one for $\sqrt[d]{n}$, as the following lemma states.
Recall that we chose to approximate $\sqrt[d]{n}$ with less precision in order to avoid updating $v.q$ too often.

\begin{lemma}\label{lemma:logn}
In a $q$-connected sorted list over the interval $[0,1)$, approximating $\log(n)$ eventually yields a value $\log(n) - \varepsilon, \varepsilon \in o(1)$ w.h.p. as long as no nodes join or leave the system.
\end{lemma}

\begin{proof}
From Lemma~\ref{q_v_approx_lemma} it follows that the real value for $\frac{1}{2}\sqrt[d]{n}$ lies within the interval $[\frac{1}{2}v.q, 2v.q]$.
The algorithm \textsc{Approximate\_log\_N} chooses the best absolute $a_k$ out of all $a_i$, $i \in \{\frac{1}{2}v.q, \ldots ,2v.q\}$ and thus is able to find the real value $k = \lfloor \frac{1}{2}\sqrt[d]{n} \rfloor$ w.h.p. (with the same argumentation as in the proof of Lemma~\ref{q_v_approx_lemma}).
Therefore, $\log((2k)^d)$ is equal to $\log(n) \pm \varepsilon$, $\varepsilon \in o(1)$.
\end{proof}

\subsection{Standard de Bruijn Edges}\label{sec:desc:sdb}
The idea to generate standard de Bruijn edges for a node $v \in V$ is as follows: In \textsc{Timeout}, $v$ sends out messages $P_0, P_1$.
We call such a message a \emph{probe}.
Probe $P_0$ stores the target location $\frac{v}{2}$ and $P_1$ stores the target location $\frac{v+1}{2}$ within itself.
We want a probe to reach the node in the system that is closest to the probe's target location.
A probe also stores $v$ itself, so that it can be sent back to $v$ immediately, once it arrives at the target node.
Recall that the two variables $v.db(1,0)$ and $v.db(1,1)$ contain the nodes that $v$ thinks are closest to the locations $\frac{v}{2}$, resp. $\frac{v+1}{2}$.
In the following, we explain the routing process for the probe $P_1$.
The routing for $P_0$ works analogously.

\begin{enumerate}
\item Forward $P_1$ to $u = v.right$.
\item Perform a standard de Bruijn hop by forwarding $P_1$ to $u.db(1,1)$.
\item Greedily forward $P_1$ via the $q$-neighborhood until some node $t$ is reached that is closest to $\frac{v+1}{2}$ based on its local view.
\item Store $t$ in $P_1$ and send the probe back to $v$, such that $v$ is able to set $v.db(1,1) = t$.
\end{enumerate}

Note that if the system has not reached a legal state yet, steps $1$ or $2$ may not be executed, since the respective variables are set to $\perp$.
In this case we proceed with step $4$, storing the most recently traversed node.
It is easy to see that once the sorted list along with the $q$-connected list has stabilized, $v.db(1,0)$ and $v.db(1,1)$ will eventually store the correct nodes.
If $v$ modifies $v.db(1,1)$, it delegates the old value for $v.db(1,1)$ away via the \textsc{BuildList} protocol.
This approach is efficient regarding the number of hops per probe as the following lemma shows:

\begin{lemma}\label{lemma:standardProbing}
Let the GDB $G$ be in a legitimate state.
Probes for standard de Bruijn edges only need $3$ hops w.h.p. to be routed from a node $v \in V$ to the node that is closest to the probe's target, namely $\frac{v+j}{2}$, $j \in \{0,1\}$.
\end{lemma}

\begin{proof}
W.l.o.g. let $j = 1$.
In a stable system, a node $v \in V$ routes the probe for the location $\frac{v+j}{2}$ to its right list neighbor $v.right$.
From there on we perform a de Bruijn hop via the edge $v.right \rightarrow \frac{v.right+1}{2}$ and end up at the node $v.right.db(1,1)$.
We show that w.h.p. $v.db(1,1)$ is contained in the $q$-neighborhood of $v.right.db(1,1)$, so the probe already reaches its destination after step $3$.
We use the following Lemma from \cite{DBLP:journals/talg/NaorW07} to compute the largest possible distance between $v.db(1,1)$ and $v.right.db(1,1)$:

\begin{lemma}[\cite{DBLP:journals/talg/NaorW07}]
After hashing $n$ nodes onto the interval $[0,1)$ the length of the longest segment is w.h.p. $\Theta(\frac{\log(n)}{n})$.
\end{lemma}

We want to compute the probability that there are more than $2 \log n$ nodes within an interval of size $\log(n)/n$.
If this probability turns out to be very low, the probability that there are more than $\sqrt[d]{n}$ nodes within an interval of size $\log n/n$ is even lower, since $p \log(n) \leq \sqrt[d]{n}$ for some constant $p$ and $n$ high enough.
As a consequence the probability that $v.db(1,1)$ is in the $q$-neighborhood of $v.right.db(1,1)$ is very high, which proves the lemma.
Let $I$ be an interval in $[0,1)$ of size $\log(n)/n$ and for all $v \in V$ let $X_v$ be a binary random variable with

\begin{equation*}
   X_v =
   \begin{cases}
     1, & \text{if node }v \text{ is in }I \\
     0, & \text{otherwise.} 
   \end{cases}
\end{equation*}

For $v \in V$ it holds $Pr[v \in I] = \frac{\log(n)}{n}$ and for $X := \sum_{v \in V} X_v$ it holds $E[X] = \log(n)$.
Following the standard Chernoff bound (Theorem~\ref{theorem:chernoff}(1)) with $\delta = 1$, we get

\begin{eqnarray*}
Pr[X \geq (1 + \delta) \cdot \mu] & = & Pr[X \geq 2 \log(n)]\\
& \leq & e^{-\log(n)/\frac{8}{3}}\\
& \leq & n^{-3/8}.
\end{eqnarray*}
\end{proof}

\subsection{General de Bruijn Edges} \label{sec:desc:general}
To establish general de Bruijn edges at each node we use a probing approach similar to that for the standard de Bruijn edges: Nodes periodically send out a probe and forward it until it arrives at the node that is closest to the probe's target location.
Since we want to avoid sending out probes for all possible general de Bruijn targets at once, we send out only one probe per \textsc{Timeout}-call for one single general de Bruijn target.
For picking the probe's target, we use a round-robin approach similar to the one for the $q$-neighborhood edges: In each call of \textsc{Timeout}, we pick $i \in \{2,\ldots,\log(v.q)+1\}$ and $j \in \{0,\ldots,2^i-1\}$ in a round-robin fashion and generate the probe $P_{i,j}$ that has the point $\frac{v+j}{2^i}$ as target location.
The result of $P_{i,j}$ has to be stored in $v.db(i,j)$.
Aside from $v$ itself, we also store $i$ and $j$ in $P_{i,j}$ since these are important for the routing approach:

\begin{enumerate}
\item Forward $P_{i,j}$ to $u = v.db(i-1,k)$, with $k = j \mod 2^{i-1}$.
\item Execute a standard de Bruijn hop: If $j \geq 2^{i-1}$ then forward $P_{i,j}$ from $u$ to $u.db(1,1)$, otherwise forward $P_{i,j}$ from $u$ to $u.db(1,0)$.
\item Greedily forward $P_{i,j}$ via the $q$-neighborhood until some node $t$ is reached that is closest to $\frac{v+j}{2^i}$ based on its local view.
\item Store $t$ in $P_{i,j}$ and send $P_{i,j}$ back to $v$, such that $v$ is able to set $v.db(i,j) = t$.
\end{enumerate}

The following lemma shows that the above approach is efficient regarding the number of hops for a single probe:

\begin{lemma}\label{lemma:probing}
Let the GDB $G$ be in a legitimate state.
Probes for general de Bruijn edges only need $3$ hops w.h.p. to be routed from a node $v \in V$ to the node that is closest to the probe's target, namely $\frac{v+j}{2^{i}}$ for $i \in \{2, \ldots ,\log(v.q)+1\}$ and $j \in \{0, \ldots ,2^{i} - 1\}$.
\end{lemma}

\begin{proof}
Analogously to the proof of Lemma~\ref{lemma:standardProbing}.
\end{proof}

Having nodes store lower-level general de Bruijn edges is not only useful in our probing approach, but also reduces the effort for $v$ when processing a $v.q$-update.
This will certainly be the case in a dynamic environment as there are nodes leaving the system resulting in $v.q$ to be halved.
As soon as $v.q$ halves at a node $v$, $v$ just drops its general de Bruijn edges on the highest level.
Without lower-level general de Bruijn edges $v$ would have to probe for a new set of general de Bruijn edges.
Similarly, in case that $v.q$ doubles, $v$ is able to use its old high-level general de Bruijn edges to probe for the general de Bruijn edges on the next higher level.

\subsection{Join and Leave}
When a new node $v$ wants to join the system at some node $u$, it just introduces itself to $u$ by calling \Call{Linearize}{$v$} on $u$.
Then $v$ is integrated into the sorted list via \textsc{BuildList}.
As soon as $v$ is in the correct spot of the sorted list, $v$ is able to generate a correct approximation $v.q$ along with its $q$-neighborhood and thus build the set of general de Bruijn edges.
A very simple approach for a node to leave the system is to 'just leave'.
Since each node is connected to its closest $\Theta(\sqrt[d]{n})$ list neighbors, when the system is in a legitimate state, the graph is guaranteed to stay weakly connected when one node leaves.
However, nodes may also leave at times at which the network has not yet reached a legitimate state.
There are already protocols that are able to safely exclude a node from the system, so we just refer the reader to \cite{DBLP:conf/sss/KoutsopoulosSS15} for a universal approach or \cite{DBLP:conf/sss/ForebackKNSS14} for an approach specifically for the sorted list.

\section{Protocol Analysis} \label{sec:analysis}
In this section we show that \textsc{BuildQDeBruijn} is self-stabilizing according to Definition~\ref{self_stabilization}.

\subsection{Convergence}
Given any weakly connected graph $G = (V, E_L \cup E_q \cup E_{dB} \cup E_{q-dB})$, we first argue that all corrupted messages stored initially in node channels are processed:

\begin{lemma} \label{lemma:convergence:0}
	Given any weakly connected graph $G = (V, E_L \cup E_q \cup E_{dB} \cup E_{q-dB})$ and a set of corrupted messages $M$ spread arbitrarily over all node channels.
	Eventually, $G$ is free of corrupted messages, while staying weakly connected.
\end{lemma}

\begin{proof}
	Any message $m$ is processed by \textsc{BuildQDeBruijn} according to the protocol description from Section~\ref{sec:desc}.
	By definition of \textsc{BuildQDeBruijn}, $G$ does not get disconnected when processing $m$.
	If $m$ is a message for \textsc{BuildList}, then we know that eventually the implicit edge represented by $m$ will be converted to an explicit edge, after $m$ has been forwarded to the right spot in the sorted list.
	If $m$ is a message for the $q$-neighborhood protocol, then the node receiving $m$ may generate an acknowledgement $m'$ when processing $m$.
	By definition of the $q$-neighborhood protocol, processing $m'$ does not lead to more acknowledgement messages, which ends the chain of corrupted messages generated by $m$.
	In case $m$ is a probe (either for standard or general de Bruijn edges), by definition of our protocol $m$ will eventually be sent back to its initial sender, at which it is then processed without generating new corrupted messages.
\end{proof}

We show the rest of the convergence proof in multiple phases: First we argue that our system converges to a sorted list from any weakly connected graph.
When the list is stable, our protocol is able to establish the $q$-neighborhood edges as well as the standard and general de Bruijn edges.

\begin{lemma} \label{lemma:convergence:I}
\textsc{BuildQDeBruijn} transforms any weakly connected graph $G = (V, E_L \cup E_q \cup E_{dB} \cup E_{q-dB})$ into a sorted list.
\end{lemma}

\begin{proof}
We show for each edge set $E_+ \in \{E_q, E_{dB}, E_{q-dB}\}$ that any weakly connected graph $G_+ = (V, E_L \cup E_+)$ converges to a graph $G_+' = (V, E_L' \cup E_+')$ such that the graph $H = (V, E_L')$ is weakly connected, because then we can apply Theorem~\ref{list:self_stabilizing} to show the convergence of the sorted list.
This suffices to show the lemma, because sub-protocols in $E_+$ do not modify variables of \textsc{BuildList} directly and do not interact with other sub-protocols in $E_+$ in terms of calling a method or modifying variables (recall Figure~\ref{fig:protocol_dependencies}).
For a graph $G = (V,E)$ we call a directed edge $(v,w) \in E$ a \emph{bridge}, if $G$ is weakly connected, but $G'=(V,E\setminus \{(v,w)\})$ is not.

\begin{lemma}\label{lemma:qNeighborhoodToList}
\textsc{BuildQDeBruijn} transforms any weakly connected graph $G = (V, E_L \cup E_q)$ into a sorted list.
\end{lemma}

\begin{proof}
For a weakly connected graph $G_q=(V,E_L \cup E_q)$ assume that there is a bridge $(v,w) \in E_q$ that splits $G_q$ into two weakly connected components $G_q' = (V',E_L' \cup E_q')$ and $G_q'' = (V'',E_L'' \cup E_q'')$ with $V' \cap V'' = \emptyset$, $E_L = E_L' \cup E_L''$ and $E_q = E_q' \cup E_q''$.
This assumption implies that $H=(V,E_L)$ is not connected yet.
The proof works analogously for the case that there are multiple edges in $E_q$ that split $G$ in two weakly connected components.
W.l.o.g. assume $v < w$ and that $G_q'$ and $G_q''$ are already weakly connected just by the list edges $E_L$.
We perform an induction over the number of nodes $N$ that lie between $v$ and $w$ on the $[0,1)$-interval and show that $H$ eventually becomes weakly connected.

We start the induction at $N = 0$ and consider the node $v$: If $v.right = \perp$, then the $q$-neighborhood edge $(v,w)$ gets downgraded to a list edge in $v$'s \textsc{Timeout} procedure (Algorithm~\ref{algo:q-neighborhood}, line~\ref{algline:downgrading_q_neighborhood_end}).
In case $v.right \neq \perp$, we know that $v.right \in V'$, since otherwise $H$ is already weakly connected, which is a contradiction to our initial assumption.
So for $v.right \neq \perp$, $v.right \in V'$ we need to consider the cases $v.right \leq v$ and $v < w < v.right$ since there are no nodes lying between $v$ and $w$ on the $[0,1)$-interval.
In case $v.right \leq v$, the \textsc{Timeout} procedure of \textsc{BuildList} calls \textsc{Linearize}($v.right$) and sets $v.right = \perp$ afterwards (Algorithm~\ref{algo:buildlist}, lines~\ref{algline:buildlist:corrupt:start}-\ref{algline:buildlist:corrupt:end}), so we get to case $v.right = \perp$.
If $v < w < v.right$, then $v$ recognizes in its \textsc{Timeout} procedure of the $q$-neighborhood protocol that there is a node $w \in v.Q$ that is closer to $v$ than $v.right$, so the $q$-neighborhood edge $(v,w)$ gets downgraded to a list edge (Algorithm~\ref{algo:q-neighborhood}, line~\ref{algline:downgrading_q_neighborhood_end}).

For the induction hypothesis, we assume that $H$ becomes weakly connected, if there exist nodes $v \in V'$ and $w \in V''$ that are connected by a $q$-neighborhood edge $(v,w)$ and there are at most $n$ nodes lying between $v$ and $w$ on the $[0,1)$-interval.

Now assume that there are $N+1$ nodes $x_1,\ldots,x_{N+1}$ lying between $v$ and $w$. Again we distinguish between two cases: If $v.right = \perp$ then $(v,w)$ gets downgraded in $v$'s \textsc{Timeout} procedure and we are done (Algorithm~\ref{algo:q-neighborhood}, line~\ref{algline:downgrading_q_neighborhood_end}).
Hence let $v.right \neq \perp$.
For $v.right \leq v$ this is solved trivially by \textsc{BuildList} like before.
Therefore, assume that $v.right > v$ with $v.right \in V'$.
Then $v.right$ is upgraded to $v.Q$ in $v$'s \textsc{Timeout} procedure (Algorithm~\ref{algo:q-neighborhood}, line~\ref{algline:promotion_q_neighborhood}), which implies that $v.Q$ contains at least the nodes $v.right$ and $w$.
Now $v$ eventually picks $q_k = w$ in its \textsc{Timeout} procedure and introduces $q_{k-1} \in V'$ to $w$ (Algorithm~\ref{algo:q-neighborhood}, line~\ref{algline:qNeighborhood_inductionEnd}).
It either holds $q_{k-1} = v.right$ or $q_{k-1} = u$ for some $u \in V'$, $u > v.right$ (recall that nodes in $v.Q$ are sorted by their hash values).
This generates a $q$-neighborhood edge $(w,q_{k-1})$ with $w$ and $q_{k-1}$ being in different weakly connected components.
Since $q_{k-1} > v$, there are less than $N+1$ nodes lying between $w$ and $q_{k-1}$ so we know from the induction hypothesis that $H$ becomes weakly connected just by list edges.
\end{proof}

\begin{lemma}\label{lemma:standardDBToList}
\textsc{BuildQDeBruijn} transforms any weakly connected graph $G = (V, E_L \cup E_{dB})$ into a sorted list.
\end{lemma}

\begin{proof}
Assume that $G = (V, E_L \cup E_{dB})$ is weakly connected, but $H = (V, E_L)$ is not.
Then $V$ can be split up into $k \geq 2$ disjoint sets $V_1,\ldots,V_k \subset V$, for which it holds that $\forall i \in \{1,\ldots,k\}\ H_i := (V_i,\{(u,v) \in E_L\ |\ u,v \in V_i\})$ is weakly connected.
Consider $V_i, V_j$ such that $i \neq j$ and $G_{i,j} := (V_i \cup V_j,\{(u,v) \in E_L \cup E_{dB}\ |\ u,v \in V_i \cup V_j\})$ is weakly connected.
$V_i,V_j$ exist, because otherwise $G$ is not connected, which contradicts our initial assumption.
To show the lemma it suffices to show that the graph \[H_{i,j} := (V_i \cup V_j,\{(u,v) \in E_L\ |\ u,v \in V_i \cup V_j\})\] eventually becomes weakly connected, because the number of disjoint sets $k$ that $V$ can be split up into reduces by one.
By repeating this process, eventually $k = 1$ and we can apply Theorem~\ref{list:self_stabilizing}.

Let $V_i, V_j$ be weakly connected components as defined above.
Assume that the nodes in $H_i, H_j$ have already formed a sorted list.
They can do that because of Theorem~\ref{list:self_stabilizing} and the fact that all other sub-protocols are not able to violate the connectivity of $H_i, H_j$. 
Then there exists at least one edge out of $E_{dB}$ that has a node in $V_i$ as source and a node in $V_j$ as target or vice versa.
W.l.o.g. assume that there exist standard de Bruijn edges $e_1,\ldots,e_m \in V_i \times V_j$, $m \geq 1$.
For simplicity let these edges be ordered according to the hash value of their source node, i.e., for $e_1=(u_1,v_1),\ldots,e_m=(u_m,v_m)$ it holds $u_1 < u_2 < \ldots < u_m$.
Furthermore, consider only the standard de Bruijn edges whose target node is stored in the variable $v.db(1,1)$ of a node $v$ (we can do that because the standard de Bruijn protocol treats the variables $v.db(1,0)$ and $v.db(1,1)$ independently from each other).
It follows that the probe for the edge $e_m$ will fail, because $u_m$ either does not have a right neighbor which triggers line 5 of Algorithm~\ref{algo:standard_de_bruijn}, or the right neighbor $u_m.right$ does not have a standard de Bruijn edge stored in $u_m.right.db(1,1)$, which triggers line 22 of Algorithm~\ref{algo:standard_de_bruijn} and thus $u_m$ downgrades $e_m$ in line 41 of Algorithm~\ref{algo:standard_de_bruijn} after it got the probe back.
This makes the graph $H_{i,j}$ weakly connected and thus shows the lemma.
\end{proof}

We can use the same argumentation as in Lemma~\ref{lemma:standardDBToList} here to show that $G = (V, E_L \cup E_{q-dB})$ eventually becomes weakly connected just by list edges:

\begin{corollary}\label{lemma:generalDBToList}
\textsc{BuildQDeBruijn} converts any weakly connected graph $G = (V, E_L \cup E_{q-dB})$ into a sorted list.
\end{corollary}

Combining Lemma~\ref{lemma:qNeighborhoodToList},~\ref{lemma:standardDBToList} and Corollary~\ref{lemma:generalDBToList} shows that our protocol converts any weakly connected graph $G = (V, E_L \cup E_q \cup E_{dB} \cup E_{q-dB})$ into a sorted list.
\end{proof}

\begin{lemma} \label{lemma:convergence:II}
\textsc{BuildQDeBruijn} converts any sorted list into a $q$-connected list.
\end{lemma}

\begin{proof}
Consider a sorted list over the $[0,1)$-interval.
In this proof it suffices to show that for a fixed value of $v.q$, the set $v.Q$ eventually contains $v$'s closest $c\cdot 2v.q$ list neighbors, because then Lemma~\ref{q_v_approx_lemma} implies that $v.q$ eventually converges to a value in $(\frac{1}{4}\sqrt[d]{n}, \sqrt[d]{n})$.
For a node $v \in V$ let $N_v = \{n_1, \ldots n_{c\cdot 2v.q}\}$ be the set of $v$'s closest $c\cdot 2v.q$ list neighbors ordered by distance to $v$, i.e., $|v-n_1| < |v - n_2| <  \ldots  < |v - n_{c\cdot 2v.q}|$.
By definition of the procedure \textsc{Introduce} (Algorithm~\ref{algo:q-neighborhood}, line \ref{algline_qNeighborhood_introduce}), we do not remove a node $n_i \in N_v$ from $v.Q$, once $v$ has added $n_i$ to $v.Q$.
The reason for this is because we only remove nodes from $v.Q$, if $|v.Q| > c\cdot 2v.q$ (Algorithm~\ref{algo:q-neighborhood}, line~\ref{algline:q_neighborhood_removeS}) and if that is the case, the nodes that are removed have a greater distance to $v$ than any node $n_i \in N_v$, so they cannot be part of $N_v$.
It suffices to show that $v$ will eventually receive a message containing node $n_i \in N_v$ for all $i \in \{1,..,c\cdot 2v.q\}$, because then $v$ will add $n_i$ to $v.Q$.

We prove this via induction over $i$: For $i = 1$ the node $n_1 \in N_v$ is the closest neighbor of $v$, so it holds $n_1 = v.left$ or $n_1 = v.right$, since the sorted list has already converged.
W.l.o.g. let $n_1 = v.left$.
Then $v$ is introduced to $n_1$ in the \textsc{Timeout} method of $v.left$ (Algorithm~\ref{algo:q-neighborhood}, line~\ref{algline:qNeighborhood_inductionStart}).

For the induction hypothesis assume that $v$ has already been introduced to nodes $n_1,...,n_i$ for a fixed $i \in \{1, \ldots ,c\cdot 2v.q\}$.
This implies $n_1,\ldots,n_i \in v.Q$.

For the induction step we show that $v$ is eventually introduced to the node $n_{i+1}$.
Since there is no node $u$ with $|v-n_i| < |v-u| < |v - n_{i+1}|$, it follows that either $n_{i+1}$ is a direct list neighbor of $v$, or there exists $j \leq i$ such that $n_j \in v.Q$ and $n_{i+1}$ are neighboring nodes.
If $n_{i+1}$ is a direct list neighbor of $v$, then we are done, since this is the same case as for the induction base.
Let $j \leq i$ such that $n_j \in v.Q$ and $n_{i+1}$ are neighboring nodes.
Then it holds either $n_j.left = n_{i+1}$ or $n_j.right = n_{i+1}$.
Node $v$ picks nodes from $v.Q$ in \textsc{Timeout} in a round-robin fashion, so $v$ will eventually pick the node $n_j$.
After calling \textsc{Introduce} on $n_j$, $n_j$ picks the direct list neighbor that is further away from $v$ (Algorithm~\ref{algo:q-neighborhood}, line~\ref{algline:qNeighborhood_response_start}-\ref{algline:qNeighborhood_response_end}), which is $n_{i+1}$, so $n_j$ responds to $v$ with calling $v \gets$~ \textsc{Introduce}($\{n_{i+1}\}, \perp$).
This implies that $v$ gets to know $n_{i+1}$ as well.
\end{proof}

\begin{lemma} \label{lemma:convergence:III}
\textsc{BuildQDeBruijn} converts any $q$-connected sorted list into a $q$-connected sorted list with standard de Bruijn edges.
\end{lemma}

\begin{proof}
For each node $v \in V$ consider the following scenario for its variable $v.db(1,0)$ (analogously for $v.db(1,1)$): If $v.db(1,0) = \perp$, then $v.db(1,0)$ is set to $v$ itself (Algorithm~\ref{algo:standard_de_bruijn}, line~\ref{algline:upgrade_stdb}).
Thus we can conclude that after every node $v$ has executed its \textsc{Timeout} procedure at least once, it holds that $v.db(1,0)$ is not equal to $\perp$.
Therefore, it holds for every node other than the outer list nodes that every probe does not fail within the first two hops, since we already are given a sorted list and the values for standard de Bruijn edges are not equal to $\perp$.
This implies that the probes find the correct node for the respective standard de Bruijn edges, since they search greedily for the node that is closest to the target point via the $q$-neighborhood after the first two hops.
For the two outer list nodes, i.e., the nodes with the highest, resp. lowest hash values, it holds the following: The left outer list node $l$ does not send out a probe for the variable $l.db(1,0)$, since $l.left = \perp$.
But as stated above, it already holds $l.db(1,0) = l$, which is the correct value for $l$'s standard de Bruijn edge $l.db(1,0)$.
The same argumentation may be applied to the outer right node of the sorted list.
\end{proof}

Following the argumentation in the proof of Lemma~\ref{lemma:convergence:III}, we get the same result for the general de Bruijn edges:

\begin{corollary} \label{lemma:convergence:IV}
\textsc{BuildQDeBruijn} transforms any $q$-connected sorted list with standard de Bruijn edges into a $q$-connected sorted list with standard and general de Bruijn edges.
\end{corollary}

Finally, we obtain the main result of this section:

\begin{theorem}[Convergence]
\textsc{BuildQDeBruijn} transforms any weakly connected graph $G = (V, E_L \cup E_q \cup E_{dB} \cup E_{q-dB})$ into a GDB.
\end{theorem}

\subsection{Closure}
In this section we prove the closure property of our protocol.
The approach is similar to the convergence proof from the last section: We prove the closure property for every single sub-protocol and combine these proofs to show the closure of the whole system.

\begin{lemma}\label{lemma:closure:I}
If the explicit edges in $G = (V, E_L)$ already form a sorted list, then they are preserved at any point in time if no nodes join or leave the system.
\end{lemma}

\begin{proof}
The lemma follows from Theorem~\ref{list:self_stabilizing}, because the \textsc{BuildList} protocol only modifies edges, if a node $v$ gets to know a node $w$ with either $v.left < w < v$ or $v < w < v.right$, which is not possible, because $v.left$ and $v.right$ already store $v$'s closest neighbors.
\end{proof}

\begin{lemma}\label{lemma:closure:II}
If the explicit edges in $G = (V, E_L \cup E_q)$ already form a $q$-connected list, then they are preserved at any point in time if no nodes join or leave the system.
\end{lemma}

\begin{proof}
We know that once $v.q \in (\frac{1}{4}\sqrt[d]{n}, \sqrt[d]{n})$, there won't happen any $v.q$-updates anymore (Lemma~\ref{q_v_approx_lemma}).
So after a node $v$ knows its closest $c\cdot 2v.q$ list neighbors, it does not add any new node to its set $v.Q$.
Since our protocol does not remove edges from $v.Q$ in a legitimate state of the system, the lemma follows.
\end{proof}

\begin{lemma}\label{lemma:closure:III}
If the explicit edges in $G = (V, E_L \cup E_q \cup E_{dB})$ already form a $q$-connected list with standard de Bruijn edges, then they are preserved at any point in time if no nodes join or leave the system.
\end{lemma}

\begin{proof}
Aside from the initial checks in \textsc{Timeout}, every node $v$ modifies its pointers $v.db(1,0)$ and $v.db(1,1)$ only in the method \textsc{Probe\_done} (see Algorithm~\ref{algo:standard_de_bruijn}) under the condition that the probe result is different from the node currently stored in $v.db(1,0)$, resp. $v.db(1,1)$.
In a stable system, all probes return the node that is closest to the probe's target position, so the result of a probe does not change: $v.db(1,0)$ and $v.db(1,1)$ already store the nodes that are closest to the points $v/2$, resp. $(v+1)/2$.
Also none of the initial checks in \textsc{Timeout} are true when $v$ already has the optimal node stored for its standard de Bruijn edges.
\end{proof}

\begin{lemma}\label{lemma:closure:IV}
If the explicit edges in $G = (V, E_L \cup E_q \cup E_{dB} \cup E_{q-dB})$ already form a $q$-connected list with standard and general de Bruijn edges, then they are preserved at any point in time if no nodes join or leave the system.
\end{lemma}

\begin{proof}
Analogously to the proof of Lemma~\ref{lemma:closure:III}.
\end{proof}

The main result of this section follows immediately from Lemma~\ref{lemma:closure:IV}:

\begin{theorem}[Closure]
If the explicit edges in $G = (V, E_L \cup E_q \cup E_{dB} \cup E_{q-dB})$ already form a GDB, then they are preserved at any point in time if no nodes join or leave the system.
\end{theorem}

\subsection{Additional Properties}
First, we show that the \textsc{BuildQDeBruijn} protocol is almost silent in a legitimate state regarding the number of messages that are sent out per node in each call of \textsc{Timeout}:

\begin{theorem}\label{message_number}
Let the GDB $G$ be in a legitimate state.
The number of messages that a single node sends out per \textsc{Timeout}-call in the \textsc{BuildQDeBruijn} protocol is constant.
\end{theorem}

\begin{proof}
We count the messages that are sent out per \textsc{Timeout}-call for a node $v \in V$: In the \textsc{BuildList} protocol, $v$ sends out two messages per \textsc{Timeout}-call to introduce itself to $v.left$ and $v.right$.
One message is sent out by $v$ to a list neighbor $u \in v.Q$ in the $q$-neighborhood sub-protocol.
Furthermore, $v$ sends out two probes per \textsc{Timeout}-call to the targets $\frac{v}{2}$ and $\frac{v+1}{2}$ via the standard de Bruijn protocol.
Finally, $v$ sends out one probe to a general de Bruijn edge.
Summing it all up, $v$ sends out $6$ messages per \textsc{Timeout}-call.
\end{proof}

Theorem~\ref{message_number} assures that nodes do not get flooded with stabilization messages.
This means incoming messages inserted into node channels will be processed quickly, which supports our main goal to deliver search requests as quickly as possible to the target node.

Next, we want to compute the expected amount of work for a single node $v \in V$.
By \emph{work} we mean the number of edges that a node has to build or redirect, when the number of participants changes.
First, we show that as soon as $v.q$ becomes stable, it only needs to be modified, if the number of nodes in our system changes quite heavily:

\begin{lemma}\label{q_v_update_lemma}
Let the GDB $G$ be in a legitimate state.
When $n$ increases by factor $2^d$ ($\Rightarrow 2^dn-n$ nodes join), then each old node $v \in V$ only needs to perform one $v.q$-update on expectation.
\end{lemma}

\begin{proof}
Let $n_{old} := n$ and $n_{new} := 2^d \cdot n_{old}$. Consider an arbitrary old node $v \in V$. Before $n$ increases, it holds $v.q \in \Theta(\sqrt[d]{n_{old}})$ (Lemma~\ref{q_v_approx_lemma}). If now $|V| \gets n_{new}$, then $q=\sqrt[d]{n_{new}} = \sqrt[d]{2^d\cdot n_{old}} = 2\cdot \sqrt[d]{n_{old}}$ increases by factor $2$, which leads to $v.q$ increasing by factor $2$. Doubling up $v.q$ is done in one $v.q$-update by our approximation approach (see Algorithm~\ref{algo:q-neighborhood}, line~\ref{algline:q_v_update}).
\end{proof}

Having to update $v.q$ only on rare occasions drastically reduces the work for a single node, because both, the $q$-neighborhood and the general de Bruijn edges are dependent on $v.q$.
The next lemma is needed for the proof of the following theorem:

\begin{lemma} \label{lemma:pointdistance}
Consider $n$ nodes that have been hashed uniformly and independently onto the $[0,1)$-interval.
Let $p \in [0,1)$ be an arbitrary point.
The expected distance between $p$ and the node that is located closest to $p$ is not higher than $\frac{1}{2n}$, when viewing the $[0,1)$-interval as a ring.
\end{lemma}

\begin{proof}
Let the point $p$ lie between two consecutive nodes $n_0$ and $n_1$.
From Fact~\ref{cor:nodedistance} we know that the expected distance between $n_0$ and $n_1$ is $\frac{1}{n}$.
The distance between $p$ and the node that is closest to $p$ is maximized, when $p$ lies exactly in the middle between $n_0$ and $n_1$, i.e., $p = \frac{n_0+n_1}{2}$.
This distance is equal to $\frac{1}{2n}$.
\end{proof}

Finally we can show an (asymptotically optimal) upper bound for the expected number of edges that need to be built or redirected for a node $v \in V$, when the number of nodes increases.

\begin{theorem}
Let the GDB $G$ be in a legitimate state.
When $n$ increases by factor $2^d$, then the number of edges that need to be built or redirected for an already existing node is in $\mathcal O(\sqrt[d]{n})$ on expectation, which is asymptotically optimal.
\end{theorem}

\begin{proof}
First, compute the amount of times a node $v \in V$ has to redirect one of its $v.q$-neighbors $q_i \in v.Q$. Since we know the expected distance between the two outer nodes $q_1$ and $q_{c\cdot 2v.q}$ of $v.Q$, we know the probability that a new node joins within $v$'s $q$-neighborhood:
\[Pr[\text{Node joins within }v's\ q\text{-neighborhood}] = \frac{c\cdot 2v.q}{n}.\]
So when $2^d\cdot n - n$ nodes join the system, we have 
\[\frac{c\cdot 2v.q}{n} \cdot (2^d \cdot n - n) = 2 \cdot c\cdot v.q \cdot (2^d - 1)\]
edge redirects on expectation for $v$'s $q$-neighborhood.

Now compute $v$'s expected number of edge redirects for its de Bruijn edges: Consider an arbitrary de Bruijn edge $\frac{v+j}{2^i}$ with $i \in \{1, \ldots ,\log(v.q)+1\}$ and $j \in \{0, \ldots 2^i-1\}$. 
The associated node to this edge is $v.db(i,j)$.
W.l.o.g. let $v.db(i,j) > \frac{v+j}{2^i}$.
In order for the edge to be redirected from $v.db(i,j)$ to a new node $w$, $w$ has to join within the interval $[\frac{v+j}{2^i} - (v.db(i,j) - \frac{v+j}{2^i}), v.db(i,j)]$.
From Lemma~\ref{lemma:pointdistance} we know that the expected size of this interval is not higher than $2\cdot \frac{1}{2n} = \frac{1}{n}$.
It follows 

\begin{eqnarray*}
& & Pr\left[\text{Join within } \left[\frac{v+j}{2^i} - \left(v.db(i,j) - \frac{v+j}{2^i}\right), v.db(i,j)\right]\right] \leq \frac{1}{n}\\
& \Rightarrow & \text{\# Edge redirects for the edge $v \rightarrow \frac{v+j}{2^i}$ is} \leq (2^d n - n)\cdot \frac{1}{n} = 2^d - 1\\
& \Rightarrow & \text{Work for $v$ is }\leq (2^d - 1)\cdot \text{Number of $v$'s de Bruijn edges}\\
&  \overset{Theorem~\ref{theorem_degree}}{=} & (2^d - 1)\cdot (4v.q-2).
\end{eqnarray*}
Putting it all together, the expected number of edge redirects for a node $v$ is equal to

\begin{eqnarray*}
& 2 \cdot c\cdot v.q \cdot (2^d - 1) + (2^d - 1)\cdot (4v.q-2) = \mathcal O(v.q) \overset{Lemma~\ref{q_v_approx_lemma}}{=} \mathcal O(\sqrt[d]{n}).
\end{eqnarray*}

Because of Lemma~\ref{q_v_update_lemma}, we know that an old node $v$ performs one update to its value $v.q$ after $2^d\cdot n-n$ nodes have joined the system.
So $v.q$ increases by factor $2$ and $v$ establishes a new level of general de Bruijn edges that contains $2v.q-1$ nodes.
Adding $2v.q-1$ to $(2^d-1)\cdot(4v.q + 2cv.q-2)$ does not change our overall costs of $\mathcal O(v.q)$.

Lastly, we argue why $\mathcal O(\sqrt[d]{n})$ edge redirects for old nodes is asymptotically optimal: If $n$ increases by factor $2^d$, the systems contains $2^d\cdot n$ nodes.
Since the degree is still $d$, Fact~\ref{fact:node_degree_lower_bound} implies that there has to be a node of degree $2\sqrt[d]{n}$, which leads to at least one old node having to at least double its degree of $\mathcal O(\sqrt[d]{n})$ (Theorem~\ref{theorem_degree}) to meet this bound.
\end{proof}

\section{Evaluation}\label{sec:evaluation}
In this section we discuss the results of our simulations.
We implemented our protocol in Java on a 64-bit machine, 2.3 GHz Intel core i3-6100U CPU (dual core) with 4GB main memory, running Windows 10.
For each value of $n \in \{1, \ldots ,500\}$ we perform $10$ runs, each starting with a different (random) weakly connected graph containing $n$ nodes.
In each run we count the number of phases that are necessary for our protocol until the diameter of the graph becomes equal to $d$.
Note that this does not necessarily imply that the system has reached a legitimate state at this point in time.
Afterwards we compute the average number of phases for each value of $n$.
A \emph{phase} is defined as follows: We call the \textsc{Timeout} procedure once on every node and wait until all messages that have been generated by these calls are processed, i.e., a phase is over if all nodes have called their \textsc{Timeout} procedure exactly one time and the channel of every node is empty.
Figure~\ref{fig:evaluation_diameter} shows the resulting graphs for values of $d \in \{2,3,4,5\}$.

\begin{figure}[ht]
	\centering
 	\includegraphics[scale=0.39]{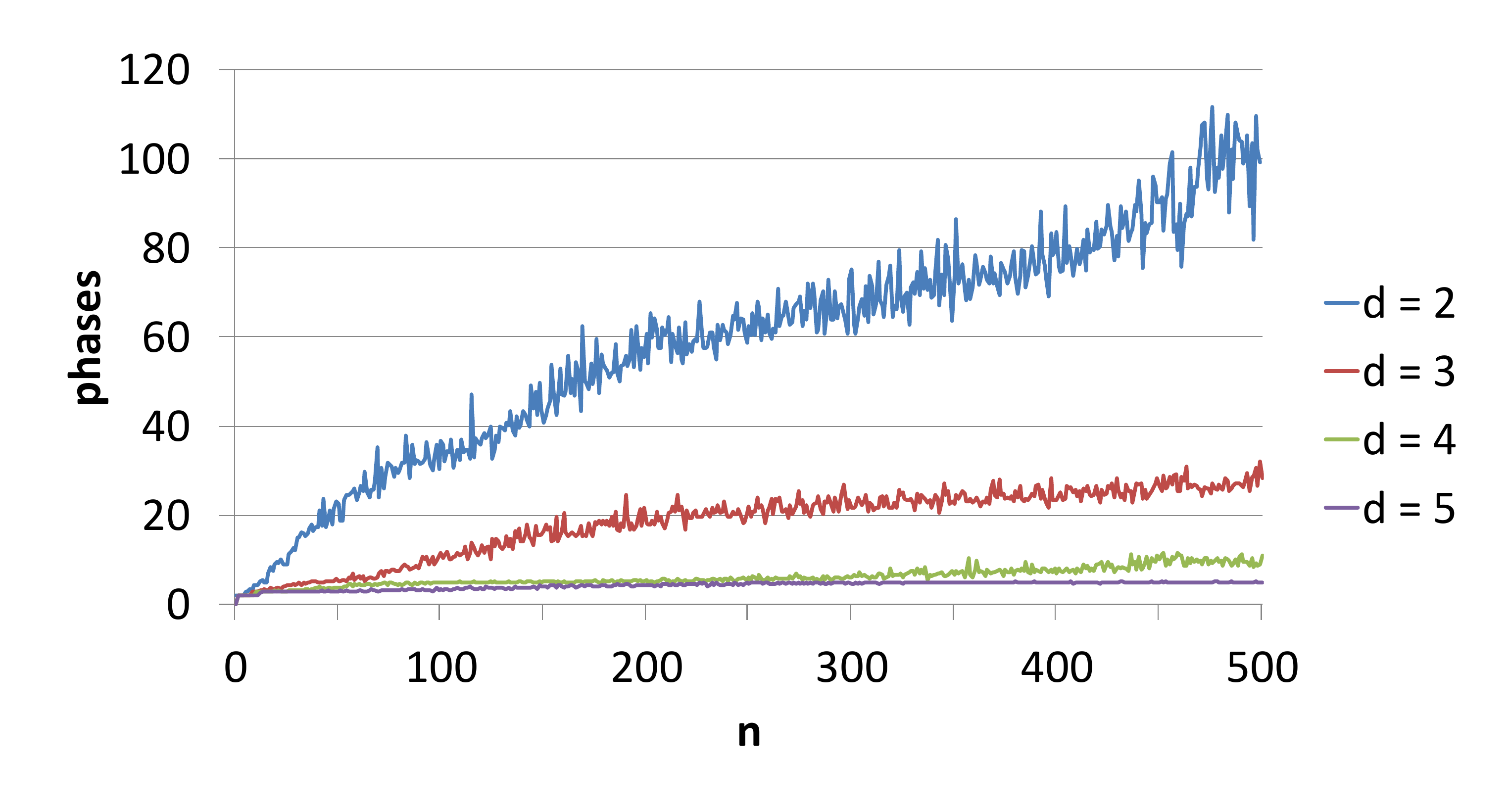}
	\caption{Average number of phases it takes our protocol to converge an initial weakly connected graph into a graph with diameter at most $d$.}
	\label{fig:evaluation_diameter}
\end{figure}

Notice that almost all our bounds that we have shown in Section~\ref{sec:analysis} are on expectation.
Because of this there exist initially weakly connected graphs that our protocol fails to converge to a graph with diameter $\leq d$.
The reason for this is because only a few entries in the distance matrix are at value $d+1$, resulting in the overall degree of the graph being equal to $d+1$.
Therefore, we allow the $q$-neighborhood of all nodes to contain more than only the $\sqrt[d]{n}$ closest nodes, by multiplying $q$ with some constant factor $c$ as described in the previous sections.
We can see that as soon as $d > 2$, the process fastens drastically.
For $d > 2$ the slope of the respective curves is very low, with the function for $d = 5$ being almost constant.
This indicates that for a large amount of nodes the system converges quickly to a practicable state such that search requests can be efficiently processed.

\section{Conclusion and future Work} \label{sec:conclusion}
We presented a new self-stabilizing protocol for the general de Bruijn graph that consists of multiple sub-protocols.
It has an advantage compared to the self-stabilizing clique in terms of the node degree, while still being able to provide constant time routing w.h.p.
Since the whole protocol is dependent on the publicly known hash function $h$ and the constant $d$, it may be an interesting task to handle nodes that use corrupted hash functions or corrupted values for $d$.


\newpage

\appendix

\section{Pseudocode}\label{appendix:pseudocode}

\begin{algorithm}[ht]
\caption{The routing algorithm $\rightarrow $ executed by node $v$}
\label{algo:routing}
\begin{algorithmic}[1]
\Procedure{DeBruijnSearch}{$t \in \mathbb{N}$, $r \in \mathbb{Z}$, $remHops \in \mathbb{N}_0$}
\If{$v.id = t$} \label{algline:search_success_1}
	\State \Return "Success!" \label{algline:search_success_2}
\EndIf
\If{$remHops > 0$ $\wedge$ $\log (q) \leq r$}
	\State $bin_t \gets $ \Call{Compute\_Bitstring}{$t, r$} \Comment{First $r$ bits of $t$'s bit string}\label{algline:routing:hop_start}
	\State $(t_1t_2 \ldots t_k)_{q} \gets$ \Call{Base\_Transform}{$bin_t, q$}  \Comment{Transform $bin_t$ to base $q$}\label{algline:routing:base_transform}
	\State Determine the edge $v \rightarrow u$ with minimal value $|u-\frac{v+t_k}{q}|$
	\State $u \gets $\Call{DeBruijnSearch}{$t$, $r - \log (q)$, $remHops-1$} \Comment{General de Bruijn hop}\label{algline:routing:hop_end}
\Else
	\State $u \gets \argmin_{(v,w) \in E_q}\{|w - h(t)|\}$ \Comment{Greedy search based on $u.Q$}\label{algline:routing:forwarding_start}
	\If{$|u-h(t)| < |v-h(t)|$}
		\State $u \gets $\Call{DeBruijnSearch}{$t$, $-1$, $0$} \Comment{Forward request to closer node}
	\Else\label{algline:search_fail_1}
		\State \Return "Failure!" \label{algline:routing:forwarding_end}
	\EndIf
\EndIf
\EndProcedure
\end{algorithmic}
\end{algorithm}

\begin{algorithm}[ht]
\caption{The \textsc{BuildList} protocol $\rightarrow $ executed by node $v$}
\label{algo:buildlist}
\begin{algorithmic}[1]
\Procedure{Timeout}{} \Comment{Introduce $v$ to $v.left$ and $v.right$}
\If{$v.right > v$} \Comment{Analogously for $v.left < v$}
	\State $v.right \gets \Call{Linearize}{v}$
\Else
	\State $\Call{Linearize}{v.right}$ \label{algline:buildlist:corrupt:start}
	\State $v.right = \perp$ \label{algline:buildlist:corrupt:end}
\EndIf
\EndProcedure
\State
\Procedure{Linearize}{$u$}
	\If{$u > v.right$} \Comment{Analogously for $u < v.left$}
		\State \Call{Delegate}{$u$}
	\EndIf
	\If{$v < u < v.right$} \Comment{Analogously for $v.left < u < v$}
		\State $u \gets \Call{Linearize}{v.right}$
		\State $v.right \gets u$
	\EndIf
\EndProcedure
\State
\Procedure{Delegate}{$u$} \Comment{Delegate $u$ to a closer node}
	\State $\bar{q} \gets \argmin_{w \in v.Q}\{|w - u|\}$
	\State $\bar{q} \gets$ \Call{Linearize}{$u$}
\EndProcedure
\end{algorithmic}
\end{algorithm}

\begin{algorithm}[ht]
\caption{The $q$-neighborhood protocol $\rightarrow $ executed by node $v$}
\label{algo:q-neighborhood}
\begin{algorithmic}[1]
\Procedure{Timeout}{}
	\For{\textbf{all} $q_r \in \{\bar{q} \in v.Q\ |\ \bar{q} > v, \bar{q} < v.right\}$}\Comment{Analogously for $v$'s left side}
		\State \Call{Linearize}{$q_r$} \Comment{Downgrade $q_r$ to \textsc{BuildList}}\label{algline:downgrading_q_neighborhood_end}
	\EndFor
	\If{$v.right \neq \perp \wedge\ v.right \not\in v.Q$} \Comment{Analogously for $v.left$}
		\State $v.Q \gets v.Q \cup \{v.right\}$	\Comment{Upgrade $v.right$ to $q$-neighborhood}\label{algline:promotion_q_neighborhood}
	\EndIf
	\State Pick $q_k \in v.Q$ in a round-robin fashion
	\If{$q_k = v.right$} \Comment{Analogously for $q_k = v.left$}
		\State $q_k \gets \Call{Introduce}{v, v}$ \label{algline:qNeighborhood_inductionStart}
	\ElsIf{$q_k > v$} \Comment{Analogously for $q_k < v$}
		\State $q_k \gets \Call{Introduce}{q_{k-1}, v}$ \label{algline:qNeighborhood_inductionEnd}
	\EndIf
	\State \Call{Approximate\_Q}{}
\EndProcedure
\State
\Procedure{Introduce}{$\tilde{Q}$, $sender$} \label{algline_qNeighborhood_introduce}
	\State $v.Q \gets v.Q \cup \tilde{Q}$
	\If{$|v.Q| > c\cdot 2v.q$}\label{algline:delegating_q_neighborhood_start}
		\For{$i = 1$ \textbf{to} $|v.Q| - c\cdot 2v.q$}
			\State $\bar{q} \gets \argmax_{q \in v.Q}\{|v-q|\}$
			\State $v.Q \gets v.Q \setminus \{\bar{q}\}$ \Comment{Remove $\bar{q}$ from $v.Q$}\label{algline:q_neighborhood_removeS}
			\State \Call{Linearize}{$\bar{q}$} \Comment{Delegate additional nodes via \textsc{BuildList}}\label{algline:delegating_q_neighborhood_end}
		\EndFor
	\EndIf
	\If{$sender \neq \perp$}
		\If{$sender < v$}\label{algline:qNeighborhood_response_start}
			\State $sender \gets \Call{Introduce}{v.right, \perp}$\Comment{Send a list neighbor back}
		\Else
			\State $sender \gets \Call{Introduce}{v.left, \perp}$\Comment{Send a list neighbor back}\label{algline:qNeighborhood_response_end}
		\EndIf
	\EndIf
\EndProcedure
\State
\Procedure{Approximate\_Q}{}
	\For{$i \in \{-\log(v.q), \ldots ,0,1\}$}
		\State $a_i  \gets |2^d \cdot |q_1-q_{2^iv.q}| - (1/(2^i\cdot v.q))^{d-1}|$ \label{algline:a_i}
	\EndFor
	\State $min \gets \argmin_{i \in \{-\log(|v.q|), \ldots ,0,1\}}\{a_i\}$ \Comment{Choose the absolute value closest to $0$}
	\If{$min \neq 0$}
		\State $v.q \gets v.q \cdot 2^{min}$ \label{algline:q_v_update}
	\EndIf
\EndProcedure
\State
\Function{Approximate\_log\_N}{}
	\For{$i \in \{\frac{1}{2}v.q, \ldots ,2v.q\}$}
		\State $a_i  \gets |2^d \cdot |q_1-q_i| - (1/i)^{d-1}|$
	\EndFor
	\State $k \gets \argmin_{i \in \{\frac{1}{2}v.q, \ldots ,2v.q\}}\{a_i\}$ \Comment{Choose the absolute value closest to $0$}
\State \Return $\lfloor \log((2k)^d) \rfloor$
\EndFunction
\end{algorithmic}
\end{algorithm}

\begin{algorithm}[ht]
\caption{The standard de Bruijn protocol $\rightarrow $ executed by node $v$}
\label{algo:standard_de_bruijn}
\begin{algorithmic}[1]
\Procedure{Timeout}{}
	\If{$v.right = \perp$ $\wedge$ $v.db(1,1) \neq \perp$} \Comment{Analogously for $v.db(1,0)$}
		\State \Call{Linearize}{$v.db(1,1)$} \Comment{Downgrade $v.db(1,1)$ to \textsc{BuildList}} \label{algline:downgrade_stdb}
	\EndIf
	\If{$v.db(1,1) = \perp \vee\ v.db(1,1) < v$}\Comment{Analogously for $v.db(1,0)$}
		\If{$v.db(1,1) \neq \perp \wedge\ v.db(1,1) < v$}
			\State \Call{Linearize}{$v.db(1,1)$}
		\EndIf
		\State $v.db(1,1) \gets v$  \Comment{Set $v$ as standard de Bruijn edge}\label{algline:upgrade_stdb}
	\EndIf
	\State $v.left \gets $ \Call{Probe}{$v,\frac{v}{2}, "leftDB"$}
	\State $v.right \gets $ \Call{Probe}{$v,\frac{v+1}{2}, "rightDB"$}
\EndProcedure
\State
\Procedure{Probe}{$sender$, $t$, $mode$}
	\If{$mode = "rightDB"$} \Comment{Analogously for $mode = leftDB$}
		\If{$v.db(1,1) = \perp$}
			\State $sender \gets $ \Call{Probe\_Done}{$t, v$} \label{algline:sdb_probe_fail}
		\Else
			\State $v.db(1,1) \gets $ \Call{Probe}{$sender, t, "dbh\_done"$} \Comment{Standard de Bruijn hop} \label{algline:sdb:standardhop}
		\EndIf
	\ElsIf{$mode = "dbh\_done"$}
		\State $u \gets \argmin_{w \in v.Q \cup \{v\}}\{|w-t|\}$ \label{algline:sdb:greedyForward}
		\If{$u \neq v$}
			\State $u \gets $ \Call{Probe}{$sender$, $t$, $"dbh\_done"$} \Comment{Search for node closest to $t$}
		\Else
			\State $sender \gets $ \Call{Probe\_Done}{$t, v$}
		\EndIf
	\EndIf
\EndProcedure
\State
\Procedure{Probe\_Done}{$t$, $result$}
	\If{$t > v$} \Comment{Analogously for $t < v$ and $v_0$}
		\If{$result \neq v.db(1,1)$ $\wedge$ $v.db(1,1) \neq \perp$}
			\State \Call{Linearize}{$v.db(1,1)$} \Comment{Downgrade old $v.db(1,1)$ to \textsc{BuildList}}\label{algline:delegating_sdB_2}
		\EndIf	
		\State $v.db(1,1) \gets result$
	\EndIf	
\EndProcedure
\end{algorithmic}
\end{algorithm}

\begin{algorithm}[ht]
\caption{The general de Bruijn protocol $\rightarrow $ executed by node $v$}
\label{algo:general_de_bruijn}
\begin{algorithmic}[1]
\Procedure{Timeout}{}
	\State Remove and downgrade general de Bruijn edges on level $l > \log(v.q)+1$
	\For{$i \gets 2 \ldots \log(v.q)+1$}\label{algline:downgrading_gdB_start}
		\For{$j \gets 0 \ldots 2^i-1$}
			\If{$v.right = \perp$ $\wedge$ $v.db(i,j) \neq \perp$ $\wedge$ $v.db(i,j) > v$}
				\State \Call{Linearize}{$v.db(i,j)$}  \Comment{Downgrade $v.left$ analogously} \label{algline:downgrading_gdB_end}
			\EndIf	
			\If{$v.db(i,j) = \perp \vee\ (\frac{v+j}{2^i} < v \wedge\ v.db(i,j) > v)$}
				\State  \Call{Linearize}{$v.db(i,j)$} \Comment{Only if $v.db(i,j) \neq \perp$}
				\State $v.db(i,j) \gets v$ \Comment{Set $v$ as general de Bruijn edge} \label{algline:gdb:promotion}
			\EndIf
		\EndFor
	\EndFor
	\State Pick $i \in \{2, \ldots ,\log(v.q)+1\}$ and $j \in \{0, \ldots ,2^i - 1\}$ in a round-robin fashion
	\State $v.db(i-1,j \mod 2^{i-1}) \gets $\Call{General\_Probe}{$v$, $\frac{v + j}{2^{i}}$, $i$, $j$, $1$}
\EndProcedure
\State
\Procedure{General\_Probe}{$sender$, $t$, $i$, $j$, $dB$}
	\If{$dB = 1$}
		\If{($v.db(1,0) = \perp$ $\wedge$ $j < 2^{i-1}$) $\vee$ ($v.db(1,1) = \perp$ $\wedge$ $j \geq 2^{i-1}$)}
			\State $sender \gets $ \Call{General\_Probe\_Done}{$v$, $i$, $j$}
			\State \Return
		\EndIf	
		\If{$j < 2^{i-1}$}
			\State $v.db(1,0) \gets $\Call{General\_Probe}{$sender$, $t$, $i$, $j$, $0$}
		\Else
			\State $v.db(1,1) \gets $\Call{General\_Probe}{$sender$, $t$, $i$, $j$, $0$}
		\EndIf	
	\Else \Comment{$dB = 0$}
		\State $u \gets \argmin_{w \in v.Q \cup \{v\}}\{|w-t|\}$
		\If{$sender = v$}
			\State $sender \gets $ \Call{General\_Probe\_Done}{$v$, $i$, $j$}
		\Else
			\State $u \gets $\Call{General\_Probe}{$sender$, $t$, $i$, $j$, $0$} \Comment{Forward probe}	
		\EndIf	
	\EndIf	
\EndProcedure
\State
\Procedure{General\_Probe\_Done}{$result$, $i$, $j$}
	\If{$result \neq v.db(i,j)$ $\wedge$ $v.db(i,j) \neq \perp$}
		\State \Call{Linearize}{$v.db(i,j)$} \Comment{Downgrade old $v.db(i,j)$ to \textsc{BuildList}}\label{algline:delegating_gdB}
	\EndIf	
	\State $v.db(i,j) \gets result$\label{algline:generaldB_establishEdge}
\EndProcedure
\end{algorithmic}
\end{algorithm}

\end{document}